\documentclass[11pt, a4paper]{article}
\usepackage{a4wide}
\usepackage{amsmath}
\usepackage{amsthm}
\usepackage{amssymb}
\usepackage{amsfonts}
\usepackage{cite}
\usepackage{enumerate}
\usepackage{array}
\date{}
\begin{document}
    \newtheorem{theorem}{Theorem}[section]
    \newtheorem{lemma}[theorem]{Lemma}
    \newtheorem{proposition}[theorem]{Proposition}
    \newtheorem{corollary}[theorem]{Corollary}
    \newtheorem{sublemma}[theorem]{Sublemma}
    \newtheorem{definition}[theorem]{Definition}

\title{The minimum distance of classical and quantum turbo-codes}

\author{Mamdouh Abbara\thanks{INRIA, Equipe Secret, Domaine de Voluceau BP 105, F-78153 Le Chesnay cedex, France.}
\and 
Jean-Pierre Tillich\raisebox{0.5ex}{\scriptsize *}}

\maketitle

\begin{abstract}
We present a theory of quantum stabilizer turbo-encoders with unbounded minimum distance. This theory is presented under a framework common to both classical and quantum turbo-encoding theory. The main conditions to have an unbounded minimum distance are that the inner seed encoder has to be recursive, and either systematic or with a totally recursive truncated decoder. This last condition has been introduced in order to obtain a theory viable in the quantum stabilizer case, since it was known that in this case the inner seed encoder could not be recursive and systematic in the same time.
\end{abstract}

\section{Introduction}

Turbo-codes are a class of codes with very good properties. Provided that the outer encoder and the inner encoder are chosen according to certain requirements, namely, that the minimum distance $d_c$ of the outer encoder is greater than or equal to $3$, and the inner encoder is recursive and catastrophic, then the random turbo-encoders built by linking these encoders by the means of an intermediate random interleaver have a minimum distance which behaves asymptotically as $N^{\frac{d_c-2}{d_c}}$ where $N$ is the size of the input; if the minimum distance of the outer encoder is equal to $2$, then the minimum distance attains $\log N$ \textit{at best} for a particular choice of the interleaver, whereas the minimum distance is bounded for a randomly chosen interleaver. On the other hand, if the inner encoder is moreover non-catastrophic, these turbo-codes show good decoding performances for the symmetric channel as well as for the erasure channel, under a regime near to the Shannon limit, and with an iterative decoding algorithm with linear complexity in $N$. The property of having an unbounded minimum distance is lost if the inner encoder is either non recursive or non systematic, and the linear iterative decoding algorithm becomes useless if the inner encoder is catastrophic, since the decoder can not extract any useful information about the local probability distribution on the state of each bit.

\subsection{The classical argument and where it fails in the quantum case}

Attempts \cite{PTO09a} were made to design quantum stabilizer turbo-codes with unbounded minimum distance and good decoding performance. However, it was found in \cite{PTO09a} that quantum inner convolutional encoders can not be recursive and non catastrophic in the same time, nor can they be, as a consequence, recursive and systematic. In this paper, we overcome the difficulty and introduce a different set of conditions under which quantum stabilizer turbo-encoders have an unbounded minimum distance. Our result is partially inspired by \cite{KU98a}. The counting argument we adapt in this paper comes by \cite{KU98a} and has the following outline. First, upper bound the number $a_{out}(w)$ of possible inputs of the outer encoder for which the output has weight $w$, and the number $a_{in}(w,\le d)$ of possible inputs of the inner encoder for which the output has weight $\le d$. Using the randomness of the interleaver, this implies an upper bound on the number of possible inputs of the turbo-encoder such that the weight of the intermediate output (after applying the outer encoder) is $w$ and the weight of the final output is $\le d$. By summing this upper bound over all possible values of $w$ one gets an upper bound on the number of possible inputs for which the output weight is less than $d$. Kahale and Urbanke show that if $d$ is $O(N^{\frac{d_c-2}{d_c}})$, such a sum tends to $0$ as $N$ tends to infinity. A very crucial hypothesis in the counting argument which can not be used in the quantum case is that one can suppose that $w \le d$ because the inner encoder is systematic. This is where a new hypothesis has to be made in the quantum case, in order to have an acceptable upper bound for the sum in the new regime where $w > d$. A possible new hypothesis we investigate in this paper is that the truncated decoder, ie the decoder where the ancillary positions are thrown, corresponding to the inner encoder is \textit{totally recursive}.

\subsection{Results obtained under the common framework}

In order to highlight the similarities and the differences between the classical and the quantum stabilizer settings, we define a formal framework to which both settings correspond as a particular case. The main idea behind this framework is to consider encoders as operations which act on a group of errors ${\cal{P}}$ ($\mathbb{F}_2$ or the factor Pauli group depending on the settings) rather on the real input states. Under this formalism, a turbo-encoder is the concatenation of three operations: an outer encoder $C_{out}^{\otimes N_{out}}$ equal to the blockwise repetition of a seed encoder $C_{out}$, a random interleaver, and an inner convolutional encoder ${C_{in}}_{N_{in}}$. We generalize the concept of \textit{recursive} and \textit{systematic} encoders, and introduce another concept of \textit{total recursiveness}. We also write the definition of the \textit{distance} $d_c$ of an encoder, and introduce the \textit{degenerate distance} $d_q$. This distance is not different from $d_c$ in the classical setting, but it has a very important role in the quantum case since $d_q$ comes instead of $d_c$ in the expression of the lower bound on the minimum distance of the turbo-encoder. Some concepts like the \textit{detours} of a convolutional encoder are more delicate and need to be redefined to comply with the quantum case. We then show the following results. The first result is only useful in the quantum setting since this is where $|{\cal{P}}| > 2$. The second result is only useful in the classical setting, since in the quantum setting the conditions on the inner encoder are impossible to satisfy (see  \cite{PTO09a}). Actually, the second result corresponds to the classical result stated by Kahale and Urbanke in \cite{KU98a} and the proof hereby follows the unpublished proof of Kahale and Urbanke.

\begin{theorem}\label{Th1}
Consider a turbo-encoder $T_N$ of size $N$, where the interleaver is randomly chosen according to a uniform distribution. Let $d_c$ and $d_q$ be respectively the distance and the degenerate distance of the outer seed encoder. If all the following conditions are realized:
\begin{itemize}
 \item $|{\cal{P}}| > 2$
 \item $d_q > 2$ (case $1$) or $d_c > d_q = 2$ (case $2$)
 \item the inner seed encoder is recursive, and the associated truncated decoder is totally recursive
\end{itemize}
then with probability going to $1$ as $N \rightarrow \infty$, the distance of the turbo-encoder is greater than:
\begin{itemize}
 \item case $1$: $N^{\alpha}$ for all $\alpha < \frac{d_q-2}{d_q}$
 \item case $2$: $\alpha \frac{\log N}{\log \log N}$ for all $\alpha < d_c - 2$
\end{itemize}
\end{theorem}

\begin{theorem}\label{Th2}
If all the following conditions are realized:
\begin{itemize}
 \item $d_q > 2$
 \item the inner seed encoder is recursive and systematic
\end{itemize}
then with probability going to $1$ as $N \rightarrow \infty$, the distance of the turbo-encoder is greater than
$N^{\alpha}$ for all $\alpha < \frac{d_q-2}{d_q}$
\end{theorem}

In the counting argument needed to prove these results, we will separate the sum into three partial sums. The first partial sum corresponds to the case where $w \le d$ as in the classical case. The second and the third partial sums correspond to the regime where $w > d$, and are separated into the case where $w$ is sublinear and  the case where $w$ is linear with respect to the input size. These three partial sums will be upper bounded by an expression tending to $0$ when the input size tends to infinity, under a set of hypothesis proper to each of them. The article is organized as follows. In Section $2$, we present the common framework. In Section $3$, we focus on the definition and some properties of the inner encoder: we introduce the formalism of convolutional morphisms, define the notions of \textit{recursiveness}, \textit{systematicity} and \textit{total recursiveness}, and define the \textit{speed} $\eta$ of a convolutional morphism. In Section $4$, we establish two upper bounds on the weight distribution $a_{in}(w,\le d)$ (and actually $a_{in}(w,d)$) of the inner convolutional encoder. The proof is delayed to Appendix $A$. In Section $5$, we establish two upper bounds on the weight distribution $a_{out}(w)$ of the outer encoder. In Section $6$ we state the upper bounds on the three partial sums needed to proof each of Theorem \ref{Th1} and Theorem \ref{Th2}, and we assemble these bounds to prove Theorems \ref{Th1} and \ref{Th2}. The proof of the partial sums is delayed to Appendix $B$ and relies on the bounds established in Sections $4$ and $5$. 

\section{The common framework}

Focusing on errors propagation under an encoding process, rather than on the code or the encoding process itself, enables to put linear classical codes and quantum stabilizer codes under one common framework, and to grasp distance properties of the underlying code. For an introduction to the theory of stabilizer codes, the reader can refer to [ref codes stabilisateurs]. Errors to consider in order to know the minimum distance of a code are bit flip errors for linear classical codes and Pauli errors for quantum stabilizer codes. In the first case this is the direct consequence of the definition of the minimum distance of a classical code; in the second case, this comes from the fact that a subspace is a quantum error correcting code for a set of errors if and only if it is a quantum error correcting code for a vector space basis of these errors. 
We will call classical and quantum setting, respectively, the model of encoding and error propagation for linear classical codes and for quantum stabilizer codes.

\subsection{Encoding protocol}

In the classical setting, a $(n,k)$ code is a linear subset of $2^k$ elements of the set of $n$ bit states $\mathbb{F}_2^n$. Such a code is the image of an encoding protocol from $\mathbb{F}_2^k$ to $\mathbb{F}_2^n$, which first appends $n-k$ ancillary bits set to $0$ and then applies a $\mathbb{F}_2$-group automorphism ${\cal{V}}$ of $\mathbb{F}_2^n$. It is thus the set of all ${\cal{V}}(\psi, 0^{n-k})$ where $\psi$ describes the set of all $k$-bits states.

In the quantum setting, a $(n,k)$ code is a $2^k$-dimensional $\mathbb{C}$-subspace of the space of $n$-qubits states ${\cal{H}}^{\otimes n}$, where ${\cal{H}} = \{\alpha | 0 \rangle + \beta | 1 \rangle, (\alpha,\beta) \in \mathbb{C}^2\}$ designates the one qubit space. It is the image of an encoding protocol from ${\cal{H}}^{\otimes k}$ to ${\cal{H}}^{\otimes n}$, consisting in the addition of $n-k$ ancillary qubits set to $ | 0 \rangle$ followed by a $\mathbb{C}$-vector space automorphism ${\cal{V}}$ of ${\cal{H}}^{\otimes n}$. Again it is the set of all ${\cal{V}}(\psi \otimes | 0 \rangle ^{n-k})$ where $\psi$ describes the set of all $k$-qubits states. ${\cal{V}}$ is more precisely a Clifford transformation, so that it stabilizes the group of Pauli transformations by conjugation.

Putting the encoding in the form of this protocol has the advantage to identify $n-k$ positions which will carry all the information about the error given by the decoding process ${\cal{V}}^{-1}$.

\subsection{Error propagation: the encoder}

Suppose that in the classical setting, an $n$-bits state $(\psi,0^{n-k})$ is subject to an error $E = E_1...E_n \in \mathbb{F}_2^n$, where $E_i = 1$ if and only if a bit flip happens at position $i$. After applying ${\cal{V}}$, we get the state ${\cal{V}}(E+\psi) = {\cal{V}}(E) + {\cal{V}}(\psi)$, which means that the effect of ${\cal{V}}$ on the set of errors $\mathbb{F}_2$ is the group automorphism ${\cal{V}}$ itself.

In the quantum setting, suppose that an $n$-qubits state $(\psi \otimes | 0 \rangle ^{n-k})$ is subject to a Pauli error $E \in G_n$. After applying ${\cal{V}}$, we get the state ${\cal{V}} E \psi = ({\cal{V}} E {\cal{V}}^{-1}){\cal{V}} \psi$, which means that the effect of ${\cal{V}}$ is to map each element $E$ of $G_n$ into ${\cal{V}} E {\cal{V}}^{-1}$, which also belongs to $G_n$ since ${\cal{V}}$ is a Clifford transformation. The operation $E \rightarrow {\cal{V}} E {\cal{V}}^{-1}$ is a group automorphism of $G_n$. Elements of $G_n$ contain a global phase in $\{1,i,-1,-i\}$ which can be omitted for $E$ as well as for ${\cal{V}} E {\cal{V}}^{-1}$ without any loss of information; this corresponds to considering the quotient Pauli group $G_n / Z(G_n)$, equal to the $n$-fold cartesian product of the quotient Pauli group of $4$ elements $G_1 / Z(G_1) = \{I,X,Y,Z\}$. The induced action of the operation $E \rightarrow {\cal{V}} E {\cal{V}}^{-1}$ on the quotient Pauli group is, again, a group automorphism.

Thus it is possible to unify both settings by writing that $\mathbb{F}_2$ and $\{I,X,Y,Z\}$ constitute a \textit{group of errors} ${\cal{P}}$ with neutral element $I$. The group of errors also needs to verify a property with respect to the decoding step and stated in the upcoming definition of the common framework. We then say that the automorphism of ${\cal{P}}^n$ engendered by the encoding process, together with the knowledge of $k$ to separate the information carriers positions from the ancillary positions, constitute an $[\![n,k]\!]$ \textit{encoder} $C$. Let us now push the comparison between the two settings a little further, by underlining the concepts of \textit{undetected} errors and \textit{harmless} errors.

\subsection{Types of errors and the distances $d_c$ and $d_q$}

Suppose now that a non trivial error affects the state obtained after the encoding operation. Since the $[\![n,k]\!]$ encoder $C$ is an automorphism, we can write this error in the form $C(E)$ with $E$ a non trivial error. When the decoding operation ${\cal{V}}^{-1}$ is performed, the state left is $E(\psi,0^{n-k})$ or $E(\psi \otimes | 0 \rangle ^{n-k})$ depending on the setting. In both settings, one sees that if the last $n-k$ positions are left unchanged by $E$, the error $C(E)$ is \textit{undetected}. This happens in the classical setting if the last $n-k$ coordinates of $E$ are equal to $0$, and in the quantum setting if the last $n-k$ coordinates of $E$ are in the set $\{I,Z\}$ of errors acting trivially on the qubit $|0 \rangle$. In the case of an undetected error, we are left with a state in the form $(\psi',0^{n-k})$ or $(\psi' \otimes | 0 \rangle ^{n-k})$ depending on the setting; if $\psi' = \psi$, the error $C(E)$ is \textit{harmless}, otherwise it is \textit{harmful}. In both settings 
, the minimum distance of the code is the minimum Hamming weight of a harmful error. We simply call this value \textit{distance} and write $d_c$, whereas the minimum Hamming weight $d_q \le d_c$ of an undetected error will be called \textit{degenerate distance}.

\subsection{Definition of the framework}

\begin{definition}
\begin{itemize}
 \item A group of errors ${\cal{P}}$ is a finite group with composition law *, neutral element $I$, which contains a strict subgroup ${\cal{Z}}$ called the group of undetected syndromes. An element of ${\cal{P}}$ is called a letter. An element of ${\cal{P}}^n$ where $n \in \mathbb{N}^*$ is called an error. The weight of an error $E \in {\cal{P}}^n$ is $\#\{i \in [\![1,n]\!], E_i \ne I\}$, where $E_i$ is the $i$-th coordinate of $E$.
 \item An $[\![n,k]\!]$ encoder, where $n \ge k$, is an isomorphism from ${\cal{P}}^k \times {\cal{P}}^{n-k}$ to ${\cal{P}}^n$.
 \item Let $E \in {\cal{P}}^n$, $E \ne I^n$, and $C$ an $[\![n,k]\!]$ encoder. For $1 \le i \le n$, let $C^{-1}(E)_i$ be the $i$-th coordinate of $C^{-1}(E)$. Then $E$ is undetected for $C$ if:
\[\forall i \in [\![k+1,n]\!], C^{-1}(E)_i \in {\cal{Z}}\]
$E$ is harmless for $C$ if it is undetected for $C$ and:
\[\forall i \in [\![1,k]\!], C^{-1}(E)_i = I\]
$E$ is harmful for $C$ if it is undetected for $C$ and:
\[\exists i \in [\![1,k]\!]: C^{-1}(E)_i \ne I\]

 \item The distance $d_c$ of $C$ is the minimum weight of a harmful error for $C$
 \item The degenerate distance $d_q$ of $C$ is the minimum weight of an undetected error for $C$
\end{itemize}
\end{definition}

Let us also introduce a few writing conventions. As presented, an error $E$ is a sequence of elements of ${\cal{P}}$. The $i$-th element of this sequence is written $E_i$ and we write $E = E_1.E_2.\,...\,.E_N$ where $N$ is the size of $E$. $E$ will also often be seen as a concatenation of errors of smaller size, each playing a particular role with respect to the encoding protocol. We have already seen the standard encoding protocol, in which the first $k$ positions of the input carry the information, the last $n-k$ positions are an ancilla, and the $n$ positions of the output carry the encoded information. Restrictions of an error to these respective positions are called information, stabilizer, and physical errors, and written $L$, $S$ and $P$, such that in a standard encoding protocol, the error affecting the input can be written $E = (L,S)$, and the output can be written $C(E) = P$. As will be presented with the convolutional encoder, some positions play the role of a memory, in which case the restriction of an error to these positions is written $M$. The weight of the information, ancilla, physical and memory parts of $E$ are written respectively $|E|_L$, $|E|_S$, $|E|_P$ and $|E|_M$. Moreover, the number of elements of $E$ in the ancilla part which belong to ${\cal{P}} \backslash {\cal{Z}}$ is called the weight of the detected syndromes and written $|E|_X$.\\

\section{The inner encoder: definitions and properties}

\subsection{Convolutional encoders, truncated convolutional decoders}

Let us now define convolutional encoders and truncated convolutional decoders. Instead of considering simply the inverse transformation of a convolutional encoder, we will remove the stabilizer part of the output, and this is why we use the word truncated. Likewise, a similarity will be drawn between these two transformations, from which we can generalize simply what happens in a recursive encoder to what happens in a totally recursive decoder. A convolutional encoder is built by using a seed $[\![n,k,m]\!]$ encoder $C$. The seed transformation related to a truncated convolutional decoder is a truncated decoder:
\begin{definition}
An $[\![n,k,m]\!]$ encoder of memory size $m$, information size $k$, and stabilizer size $n-k$, is an isomorphism from ${\cal{P}}^{m} \times {\cal{P}}^k \times {\cal{P}}^{n-k}$ to ${\cal{P}}^{n} \times {\cal{P}}^{m}$. The truncated decoder of an $[\![n,k,m]\!]$ encoder $C$, is the application $\bar{C}$ from ${\cal{P}}^{m} \times {\cal{P}}^{n}$ to ${\cal{P}}^{k} \times {\cal{P}}^{m}$ obtained by truncating the output of $C^{-1}$ to its information and memory parts, more precisely, for all $(M',P) \in {\cal{P}}^{m} \times {\cal{P}}^{n}$, if:
\[C^{-1}(P,M') = (M,L,S)\]
then:
\[\bar{C}(M',P) = (L,M)\]
\end{definition}

One very important thing to notice here is that an encoder and a truncated decoder are both particular cases of a morphism from ${\cal{P}}^{m} \times {\cal{P}}^k \times {\cal{P}}^{s}$ to ${\cal{P}}^{n} \times {\cal{P}}^{m}$. For an encoder, we have $n \ge k$ and $s = n-k$, whereas for a truncated decoder, we have $s=0$ and the roles of $k$ and $n$ are switched. Let us thus say that a morphism from ${\cal{P}}^{m} \times {\cal{P}}^k \times {\cal{P}}^{s}$ to ${\cal{P}}^{n} \times {\cal{P}}^{m}$ is an $[\![n,k,s,m]\!]$ \textit{morphism}. In this case, an $[\![n,k,m]\!]$ encoder is an $[\![n,k,n-k,m]\!]$ morphism, and the truncated decoder associated to an $[\![n,k,m]\!]$ encoder is an $[\![k,n,0,m]\!]$ morphism.\\

A convolutional encoder $C_N$ of size $N$ and parameters $(m,k,n)$ is an isomorphism from ${\cal{P}}^m \times {\left({\cal{P}}^k \times{\cal{P}}^{n-k}\right)}^{\otimes N}$ to ${({\cal{P}}^n)}^{\otimes N} \times {\cal{P}}^m$. An input error is the concatenation of a memory error with $N$ alternations of one information and one stabilizer error, and an output error is the concatenation of $N$ physical errors ended by a memory error. The convolutional encoding is done by repeating $N$ times a given $[\![n,k,m]\!]$ encoder $C$, which acts each time on one memory error, one information error and one stabilizer error, and produces one physical error and one memory error. Thus at the $i$-th step, the intermediate error is the concatenation of $i-1$ physical errors, one memory error, and $N-i$ couples of information and stabilizer errors. Thus formally, we define a convolutional encoder by introducing $N$ intermediate transformations acting on a total of $N+1$ intermediate groups. A truncated convolutional decoder $\bar{C}_N$ is a morphism from ${\cal{P}}^m \times {({\cal{P}}^n)}^{\otimes N}$ to ${({\cal{P}}^k)}^{\otimes N} \times {\cal{P}}^m$, obtained by applying successively a truncated decoder $\bar{C}$. Its input is a memory error concatenated with $N$ physical errors, and its output is the concatenation of $N$ information errors ended by a memory error. Let us define a convolutional encoder and a truncated convolutional decoder, after defining, for each of them, the $N+1$ intermediate groups and the $N$ intermediate transformations.

\begin{definition}
Let $C$ be an $[\![n,k,m]\!]$ encoder and let $N \in \mathbb{N^*}$. For $1 \le i \le N+1$, let:
\[{\cal{P}}_{i,N} = {({\cal{P}}^n)}^{\otimes i-1} \times {\cal{P}}^m \times {({\cal{P}}^k \times{\cal{P}}^{n-k})}^{\otimes N-i+1}\]
and let:
\[{\cal{P}}_{-i,N} = {({\cal{P}}^k)}^{\otimes N-i+1} \times {\cal{P}}^m \times {({\cal{P}}^n)}^{\otimes i-1}\]
For $1 \le i \le N$, let $C^i_N$ be the isomorphism from ${\cal{P}}_{i,N}$ to ${\cal{P}}_{i+1,N}$ defined by:
\[C^i_N (P_1,..., P_{i-1},M_{i-1},L_i,S_i,...,L_N,S_N) = (P_1,..., P_i,M_i,L_{i+1},S_{i+1},...,L_N,S_N)\]
where
\[(P_i,M_i) = C(M_{i-1},L_i,S_i)\]
Let also $\bar{C}^i_N$ be the morphism from ${\cal{P}}_{-i-1,N}$ to ${\cal{P}}_{-i,N}$ defined by:
\[\bar{C}^i_N (L_N,..., L_{i+1},M_i,P_i,...,P_1) = (L_N,...,L_i,M_{i-1},P_{i-1},...,P_1)\]
where
\[(L_i,M_{i-1}) = \bar{C}(M_i,P_i)\]

The convolutional encoder $C_N$ is the isomorphism from ${\cal{P}}_{1,N}$ to ${\cal{P}}_{N+1,N}$ defined by:
\[C_N = C^N_N \circ C^{N-1}_N \circ ... \circ C^1_N\]

The truncated convolutional decoder $\bar{C}_N$ is the morphism from ${\cal{P}}_{-N-1,N}$ to ${\cal{P}}_{-1,N}$ defined by:
\[\bar{C}_N = \bar{C}^1_N \circ ... \circ \bar{C}^{N-1}_N \circ \bar{C}^N_N\]
\end{definition}

Notice that if one runs the truncated convolutional decoder $\bar{C}_N$ on the output of the convolutional encoder $C_N$, all the intermediate errors $M_i$ during the truncated convolutional decoding stage correspond to the $M_i$ of the convolutional encoding stage, and the information errors obtained at the end of the truncated decoding are equal to the initial information errors before the encoding. The memory errors $M_i$ for $1 \le i \le N$ will play an important role later; let us call $M_i$ the $i$th intermediate memory error. Also, let us call $i$th intermediate sequence the sequence $(P_1,..., P_i,M_i,L_{i+1},S_{i+1},...,L_N,S_N)$ obtained at the step $i$ of the convolutional encoding. A convolutional encoder and a truncated convolutional decoder are both a case of a \textit{convolutional morphism} defined as follows.

\begin{definition}
Let ${\cal{C}}$ be an $[\![n,k,s,m]\!]$ morphism and let $N \in \mathbb{N^*}$. For $1 \le i \le N+1$, let:
\[{\cal{P}}_{i,N} = {({\cal{P}}^n)}^{\otimes i-1} \times {\cal{P}}^m \times {({\cal{P}}^k \times{\cal{P}}^s)}^{\otimes N-i+1}\]
For $1 \le i \le N$, let ${\cal{C}}^i_N$ be the isomorphism from ${\cal{P}}_{i,N}$ to ${\cal{P}}_{i+1,N}$ defined by:
\[{\cal{C}}^i_N (P_1,..., P_{i-1},M_{i-1},L_i,S_i,...,L_N,S_N) = (P_1,..., P_i,M_i,L_{i+1},S_{i+1},...,L_N,S_N)\]
where
\[(P_i,M_i) = {\cal{C}}(M_{i-1},L_i,S_i)\]
The convolutional morphism ${\cal{C}}_N$ is the isomorphism from ${\cal{P}}_{1,N}$ to ${\cal{P}}_{N+1,N}$ defined by:
\[{\cal{C}}_N = {\cal{C}}^N_N \circ {\cal{C}}^{N-1}_N \circ ... \circ {\cal{C}}^1_N\]
\end{definition}

It is useful to separate the physical part and the memory part of an error output by a convolutional morphism:

\begin{definition}
Let $N \in \mathbb{N}^*$. $\pi_N$ and $\mu_N$ are the maps from ${\cal{P}}_{1,N}$ to respectively ${({\cal{P}}^n)}^{\otimes N}$ and ${\cal{P}}^m$, such that ${\cal{C}}_N = (\pi_N, \mu_N)$.
\end{definition}

We also define a convolutional encoder $C_\infty$ and truncated decoder $\bar{C}_{\infty}$ for inputs of infinite size. They are both a particular case of the following convolutional morphism.

\begin{definition}
The convolutional morphism ${\cal{C}}_\infty$ is the morphism from ${\cal{P}}^m \times {\left({\cal{P}}^k \times{\cal{P}}^s \right)}^{\otimes \mathbb{N}}$ to ${({\cal{P}}^n)}^{\otimes \mathbb{N}}$, such that for any infinite input sequence $E = (M,L_1,S_1,L_2,S_2,...)$, and for all $i \in \mathbb{N}^*$, the sequence of the first $i$ errors of $C_\infty (E)$ is equal to $\pi_i(M,L_1,S_1,...,L_i,S_i)$.\\
\end{definition}

The following lemma is a direct consequence of the convolutional construction:\\

\begin{lemma}
For all $1 \le i \le N$, and for all input sequences $(M,L_1,S_1,...,L_N,S_N)$:
\begin{align*}
C_N (M,L_1,S_1,...,L_N,S_N) = & \pi_i(M,L_1,S_1,...,L_i,S_i).\\
& C_{N-i}(\mu(M,L_1,S_1,...,L_i,S_i), L_{i+1},S_{i+1},...,L_N,S_N)
\end{align*}
\end{lemma}

\subsection{Recursive, systematic encoders, and totally recursive decoders}
We define here the notion of a \textit{recursive} encoder and a \textit{systematic} encoder, which are an extension of the already existing notions known for classical convolutional encoders. We also present the new notion of a \textit{total recursive} decoder. The word \textit{total} is used to emphasize the fact that for a decoder to be \textit{totally recursive}, no condition is required about the stabilizer weight of the output sequence of the decoder.

\begin{definition}
An $[\![n,k,s,m]\!]$ morphism ${\cal{C}}$ is recursive if:
\[\forall E \in {\cal{P}}^m \times {\left({\cal{P}}^k \times{\cal{P}}^s \right)}^{\otimes \mathbb{N}}, |E|_M = 0, |E|_L = 1\ \text{and}\ |E|_X = 0 \Rightarrow |{\cal{C}}_{\infty}(E)| = \infty\]
An $[\![n,k,m]\!]$ encoder $C$ is \textit{recursive}, or a truncated decoder $\bar{C}$ is \textit{totally recursive}, if the $[\![n,k,s,m]\!]$ morphism they correspond to is recursive. In other words, an $[\![n,k,m]\!]$ encoder $C$ is recursive if:
\[\forall E \in {\cal{P}}^m \times {\left({\cal{P}}^k \times{\cal{P}}^{n-k}\right)}^{\otimes \mathbb{N}}, |E|_M = 0, |E|_L = 1\ \text{and}\ |E|_X = 0 \Rightarrow |C_{\infty}(E)| = \infty\]
An $[\![n,k,m]\!]$ truncated decoder $\bar{C}$ is totally recursive if:
\[\forall E \in {\cal{P}}^m \times {({\cal{P}}^n)}^{\otimes \mathbb{N}}, |E|_M = 0\ \text{and}\ |E|_P = 1 \Rightarrow |\bar{C}_{\infty}(E)| = \infty\]

An $[\![n,k,m]\!]$ encoder $C$ is systematic if:
\[\forall N \in \mathbb{N}, \forall E \in {\cal{P}}_{1,N}, |C_N(E)| \ge |E|_L\]

\end{definition}

An equivalent characterization of recursiveness will come after defining the sets $\mathbb{M}_0$ and $\mathbb{M}_1$.


\subsection{Memory errors, speed, and characterisation of recursive encoders}

In this subsection, let ${\cal{C}}$ be an $[\![n,k,s,m]\!]$ morphism. Memory errors have two kinds of behaviour. Consider an infinite input sequence $E$ composed by concatenating a memory error $M$ and infinitely many couples of one information error and one stabilizer error, where all the information errors are equal to $I$, and all the stabilizer errors belong to ${\cal{Z}}^{s}$. Then depending on $M$, the corresponding output will either have an infinite weight independently of the sequence of stabilizer errors, or there will exist a stabilizer sequence for which the output weight is finite. This is what we intend to describe in the definition of these two sets:

\begin{definition}
\begin{eqnarray*}
\mathbb{M}_0 & = & \{ M \in {\cal{P}}^m / \exists (S_1,S_2,...) \in {{\cal{Z}}^{s}}^{\otimes \mathbb{N}} / |{\cal{C}}_{\infty}(M,I,S_1,I,S_2,...)| < \infty \}\\
\mathbb{M}_1 & = & \{ M \in {\cal{P}}^m / \forall (S_1,S_2,...) \in {{\cal{Z}}^{s}}^{\otimes \mathbb{N}}, |{\cal{C}}_{\infty}(M,I,S_1,I,S_2,...)| = \infty \}
\end{eqnarray*}
\end{definition}

Let us also define a set useful for the characterization of a recursive morphism, the set of memory errors $\mathbb{I}$, accessible by starting with the $I$ memory error and applying ${\cal{C}}$ a finite number of times on an input of information weight $0$ and detected syndrome weight $0$:
\begin{definition}
$\mathbb{I} = \{\ \mu_N (I,I,S_1,...,I_,S_N), N \in \mathbb{N}, (S_1,...,S_N) \in {({\cal{Z}}^{s})}^N \ \}$
\end{definition}

The goal of the next two lemmas is to prove that if a memory error belongs to $\mathbb{M}_1$, then the weight of a finite output ${\cal{C}}_N (I,I,S_1,...,I_,S_N)$ will be proportional to its length $N$. The constant of proportionality is in the form $1/\eta$ where the integer $\eta$ depends on ${\cal{C}}$.

\begin{lemma}
There exists a positive integer $\eta$ such that:
\[\forall M \in \mathbb{M}_1, \forall (S_1,...,S_{\eta}) \in {({\cal{Z}}^{s})}^{\otimes \eta}, |\pi_{\eta}(M,I,S_1,...,I,S_{\eta})| \geq 1\]
The smallest such $\eta$ is called the speed of the morphism ${\cal{C}}$.
\end{lemma}

\begin{proof}
We prove the lemma by contradiction by supposing that:
\[\forall \eta \in \mathbb{N}^*, \exists M_{\eta} \in \mathbb{M}_1, \exists (S_1^{\eta},...,S_{\eta}^{\eta}) \in {({\cal{Z}}^{s})}^{\otimes \eta} / \ |\pi_{\eta}(M,I,S_1^{\eta},...,I,S_{\eta}^{\eta})| = 0\]

Since the sequence $(M_{\eta})_{\eta \in \mathbb{N}}$ has its values in the finite set $\mathbb{M}_1$, there exists $M \in \mathbb{M}_1$ such that $|\{\eta \in \mathbb{N}^* / \ M_{\eta} = M \}| = \infty$.\\

Now let us define a sequence of stabilizer errors $(S_i)_{i \in \mathbb{N}^*} \in {({\cal{Z}}^{s})}^{\mathbb{N}}$ such that:
\[ \forall i \in \mathbb{N}^*, |\{ \eta \ge i : (S_1^{\eta},...,S_i^{\eta}) = (S_1,...,S_i) \}| = \infty\]

\begin{itemize}
\item The sequence $(S_1^{\eta})_{\eta \in \mathbb{N}}$ has its values in the finite set ${\cal{Z}}^{s}$. Thus, define $S_1$ as an element of ${\cal{Z}}^{s}$ such that:
\[|\{ \eta \ge 1 : S_1^{\eta} = S_1 \}| = \infty\]
\item Suppose that for a given $i \in \mathbb{N}^*$, we have defined $i$ stabilizer errors $S_1,...,S_i$, such that for all $j \le i$:
\[|\{ \eta \ge j : (S_1^{\eta},...,S_j^{\eta}) = (S_1,...,S_j) \}| = \infty\]
Let $\mathbb{N}_i = \{ \eta \ge i+1 : (S_1^{\eta},...,S_i^{\eta}) = (S_1,...,S_i) \}$, which by hypothesis is infinite. Since the sequence $(S_{i+1}^{\eta})_{\eta \in \mathbb{N}_i}$ has its values in the finite set ${\cal{Z}}^{s}$, we can define $S_{i+1}$ as an element of ${\cal{Z}}^{s}$ such that:
\[|\{ \eta \in \mathbb{N}_i: S_{i+1}^{\eta} = S_{i+1} \}| = \infty\]
which implies:
\[|\{ \eta \ge i+1 : (S_1^{\eta},...,S_{i+1}^{\eta}) = (S_1,...,S_{i+1}) \}| = \infty\]
\end{itemize}

Now let us use this sequence $(S_i)_{i \in \mathbb{N}^*} \in {({\cal{Z}}^{s})}^{\mathbb{N}}$ to produce a contradiction. Let $i \in \mathbb{N}^*$, and let $\eta \ge i$ such that $(S_1^{\eta},...,S_i^{\eta}) = (S_1,...,S_i)$. Then:
\[ |\pi_i(M,I,S_1,...,S_i)| = |\pi_i(M,I,S_1^{\eta},...,S_i^{\eta})| = 0\]

Since this is true for all $i \in \mathbb{N}^*$, $|{\cal{C}}_{\infty}(M,I,S_1,I,S_2,...)| = 0$, which contradicts the fact that $M \in \mathbb{M}_1$.
\end{proof}

\begin{lemma}
For all $N > 0$, and for all $(S_1,...,S_N) \in {\left({\cal{Z}}^{s}\right)}^{\otimes \mathbb{N}}$ the function:
\begin{eqnarray*}
{\cal{P}}^m & \rightarrow & {\cal{P}}^m\\
M & \mapsto & \mu_N(M,I,S_1,...,I,S_N)
\end{eqnarray*}
stabilizes the set $\mathbb{M}_1$.
\end{lemma}

\begin{proof}
Let $ N > 0$ and $(S_1,...,S_N) \in {({\cal{Z}}^{s})}^{\otimes N}$. Let $M \in \mathbb{M}_1$, and let $(P,M') = C_N(M,I,S_1,...,I,S_N)$.\\
For all $(S_{N+1},S_{N+2},...) \in {({\cal{Z}}^{s})}^{\otimes \mathbb{N}}$:
\[{\cal{C}}_{\infty}(M,I,S_1,I,S_2,...) = (P,\ {\cal{C}}_{\infty}(M',I,S_{N+1},I,S_{N+2},...))\]

And since $M \in \mathbb{M}_1$, $|{\cal{C}}_{\infty}(M',I,S_{N+1},I,S_{N+2},...)| = \infty$.\\

Thus $M' \in \mathbb{M}_1$.
\end{proof}

As a consequence:

\begin{lemma}
For all $N > 0$, $M \in \mathbb{M}_1$, and $(S_1,...,S_N) \in {({\cal{Z})}^{s}}^{\otimes N}$:
\[|\pi_N(M,I,S_1,...,I,S_N)| \ge \lfloor N/ \eta \rfloor\]
\end{lemma}

\begin{proof}
Let $M \in \mathbb{M}_1$. We prove the lemma by recursion on $p = \lfloor N/ \eta \rfloor$.\\

For $p = 0$, the property is clearly true.\\

Now suppose the property is true for a given $p-1$, where $p \ge 1$.\\

Let $N > 0$ such that $\lfloor N/ \eta \rfloor = p$, and let $(S_1,...,S_N) \in {({\cal{Z})}^{s}}^{\otimes N}$. The following relation holds:
\[\pi_N(M,I,S_1,...,I,S_N) = \pi_{\eta}(M,I,S_1,...,I,S_{\eta}).\ \pi_{N - \eta}(\mu_{\eta}(M,I,S_1,...,I,S_{\eta}),I,S_{\eta + 1},...,I,S_N)\]
Since $M \in \mathbb{M}_1$, $\mu_{\eta}(M,I,S_1,...,I,S_{\eta}) \in \mathbb{M}_1$, which implies by the recursion hypothesis:
\[|\pi_{N - \eta}(\mu_{\eta}(M,I,S_1,...,I,S_{\eta}),I,S_{\eta + 1},...,I,S_N)| \ge p - 1\]
And since $ |\pi_{\eta}(M,I,S_1,...,I,S_{\eta})| \geq 1$, the property is also true for $p$.
\end{proof}

Another consequence is the following characterization of recursive morphisms:

\begin{lemma}
${\cal{C}}$ is recursive if and only if for all $(M,L,S) \in \mathbb{I} * {\cal{P}}^k * {\cal{Z}}^{s}$ with $|L| = 1$, $\mu(M,L,S) \in \mathbb{M}_1$.
\end{lemma}

\begin{proof}
Let $E$ be an infinite input sequence of memory weight $0$, information weight $1$, and such that all its stabilizer letters are in ${\cal{Z}}^{s}$. Let $K$ be the position of the only information error $L$ of weight $1$. $E$ is in the form:
\[E = (I,I,S_1,...,I,S_{i-1},L,S_i,I,S_{i+1},I,S_{i+2},...)\]

Let $(P,M) = {\cal{C}}_{i-1}(I,I,S_1,...,I,S_{i-1})$. Then:
\begin{eqnarray*}
{\cal{C}}_{\infty}(E) & = & P.\ {\cal{C}}_{\infty}(M,L,S_i,I,S_{i+1},I,S_{i+2},...) \\
& = & P.\ \pi(M,L,S_i).\ {\cal{C}}_{\infty}(\mu(M,L,S_i),I,S_{i+1}I,S_{i+2},...)\\
\end{eqnarray*}

Thus, $|{\cal{C}}_{\infty}(E)| = \infty$ for all $E$ which verify the assumptions above, if and only if $\mu(M,L,S_i) \in \mathbb{M}_1$ for all $i \ge 1$, $L$ of weight $1$, $S_i$ in ${\cal{Z}}^s$, and $M$ in the form $\mu_{i-1}(I,I,S_1,...,I,S_{i-1})$ where $(S_1,...,S_{i-1}) \in {({\cal{Z}}^s)}^{i-1}$, or in other words for all $M \in \mathbb{I}$.

\end{proof}

\section{Weight distribution of the inner encoder: upper bounds in the recursive and the totally recursive cases} 

In this section, we will establish two upper bounds to the weight distribution of a convolutional encoder. What is meant precisely by the weight distribution is the number of possible (memory and information) subsequences $(M,L_1,...,L_N)$ of an input sequence $E = (M,L_1,S_1,...,L_N,S_N)$, for a given weight $w$ of the input subsequence and a given weight $d$ of the output sequence. The first upper bound comes provided that the seed encoder $C$ is recursive, and the second one comes provided that the seed truncated decoder $\bar{C}$ is totally recursive. We will establish an upper bound for a recursive morphism, and derive the two desired results as a corollary.

\subsection{Trace and detours}

Before going to the point of the upper bounds, we first need to exhibit a characteristic behaviour of recursive morphisms; mainly, we will first show that the \textit{trace} of an input sequence of a convolutional morphism is a concatenation of \textit{detours}. By introducing the concept of \textit{trace}, this phenomenon of \textit{detours} is a generalisation of the phenomenon described by \cite{KU98a} in the case of a classical convolutional encoder based on a recursive seed encoder. Consider a recursive $[\![n,k,s,m]\!]$ morphism ${\cal{C}}$. Let $N$ be a positive integer, and consider $E = (M,L_1,S_1,...,L_N,S_N) = E_1\ ...\ E_{\cal{N}}$ an input sequence to the convolutional morphism ${\cal{C}}_N$, where ${\cal{N}} = m + N (k+s)$. Thus we have $M \in {\cal{P}}^m$, and for all $i \in [\![1,n]\!], L_i \in {\cal{P}}^k$ and $S_i \in {\cal{Z}}^s$. Let $(P_1,...,P_N,M') = {\cal{C}}_N(E)$ be the output sequence.\\

Let $w_M = |E|_M = |M|$ and $w_L = |E|_L = |(L_1,...,L_N)|$. By looking at the information part $(L_1,...,L_N)$ as a sequence of $k N$ letters, label by $1 \le p_1 < ... < p_{w_L} \le {k N}$ the positions of the $w_L$ non identity letters, and for each $i \in [\![1;w_L]\!]$, let $N_i = \lceil p_i / k \rceil$ so that the $i$th non identity letter is part of the error $L_{N_i}$.

Define the $i$th truncature of $E$ as the sequence:
\[E_{\backslash i} = (M,L_1,S_1,...,L_{N_{i-1}}, S_{N_{i-1}}, L'_{N_i}, S_{N_i})\]
where $(L_1,...,L_{N_{i-1}},L'_{N_i})$ is obtained by replacing all the letters in $(L_1,...,L_{N_i})$ after the position $p_i$ by an $I$. Let us also set $p_0 = 0$, $N_0 = 0$ and $E_{\backslash 0} = M$. Define $M_i$ as the memory error of the output when running ${\cal{C}}_{N_i}$ on $E_{\backslash i}$:
\[M_i = \mu_{N_i} (E_{\backslash i})\]
Set also $M_0 = M$. This enables to define the trace of $E$:
\begin{definition}
The trace of $E$ is the sequence $(b_0,...,b_{w_L})$ of elements of $\{0,1\}$ such that for all $i \in [\![0;w_L]\!]$, $M_i \in \mathbb{M}_{b_i}$.
\end{definition}

The following result holds as a direct consequence of this definition:
\begin{lemma}
$b_i = 1$ if and only if for all sequences $(\tilde{S}_K)_{K \in \mathbb{N^*}}$ of stabilizer errors in ${\cal{Z}}^s$:
\[|C_{\infty}(E_{\backslash i}.\ (I,\tilde{S}_1,I,\tilde{S}_2,...))| = \infty\]
\end{lemma}

\begin{proof}
$b_i = 1$ if and only if $M_i \in \mathbb{M}_1$, thus if and only if for all sequences $(\tilde{S}_K)_{K \in \mathbb{N^*}}$ of stabilizer errors in ${\cal{Z}}^s$:
\[|C_{\infty}(M_i,I,\tilde{S}_1,I,\tilde{S}_2,...)| = \infty\]
The relations:
\[C_{\infty}(E_{\backslash i}.\ (I,\tilde{S}_1,I,\tilde{S}_2,...)) = \pi_{N_i}(E_{\backslash i}).\ C_{\infty}(M_i,I,\tilde{S}_1,I,\tilde{S}_2,...)\]
if $i \ge 1$, and if $i = 0$:
\[C_{\infty}(E_{\backslash 0}.\ (I,\tilde{S}_1,I,\tilde{S}_2,...)) = C_{\infty}(M_0,I,\tilde{S}_1,I,\tilde{S}_2,...)\]
show that this happens if and only if $|C_{\infty}(E_{\backslash i}.\ (I,\tilde{S}_1,I,\tilde{S}_2,...))| = \infty$.
\end{proof}

This implies the following property about the trace of $E$:\\

\begin{lemma}
For all $i$ in $[\![0, w_L-1]\!]$, if $b_i = 0$  then $b_{i+1} = 1$.
\end{lemma}

\begin{proof}
Let us suppose that $b_i = 0$. Then there exists a sequence $(S'_K)_{K \in \mathbb{N^*}}$ of stabilizer errors in ${\cal{Z}}^s$ such that:
\[|C_{\infty}(E_{\backslash i}.\ (I,S'_1,I,S'_2,...))| < \infty\]
Let $(\tilde{S}_K)_{K \in \mathbb{N^*}}$ be a sequence of stabilizer errors in ${\cal{Z}}^s$. Then the two sequences
\[E_{\backslash i}.\ (I,S'_1,I,S'_2,...)\]
and
\[E_{\backslash i+1}.\ (I,\tilde{S}_1,I,\tilde{S}_2,...)\]
do not differ in their memory part, differ in their information parts by only the $p_{i+1}$th letter of $(L_1,...,L_N)$, and differ in their stabilizer parts by a sequence of errors in ${\cal{Z}}^s$. Thus one can write:
\[E_{\backslash i+1}.\ (I,\tilde{S}_1,I,\tilde{S}_2,...) = E_{\backslash i}.\ (I,S'_1,I,S'_2,...) * \Delta E\]
where the sequence $\Delta E$ verifies $|\Delta E|_M = 0$, $|\Delta E|_L = 1$ and $|\Delta E|_X = 0$. Since ${\cal{C}}$ is recursive:
\[|{\cal{C}}_{\infty}(\Delta E)| = \infty\]
And since ${\cal{C}}_{\infty}$ is a morphism:
\[{\cal{C}}_{\infty}(E_{\backslash i+1}.\ (I,\tilde{S}_1,I,\tilde{S}_2,...)) = {\cal{C}}_{\infty}(E_{\backslash i}.\ (I,S'_1,I,S'_2,...)) * {\cal{C}}_{\infty}(\Delta E)\]
We deduce that:
\[|{\cal{C}}_{\infty}(E_{\backslash i+1}.\ (I,\tilde{S}_1,I,\tilde{S}_2,...))| = \infty \]
And since this is true for all sequences $(\tilde{S}_K)_{K \in \mathbb{N^*}}$ of stabilizer errors in ${\cal{Z}}^s$, $b_{i+1} = 1$.
\end{proof}

\begin{definition}
A detour is a finite and non empty sequence of $1$, potentially completed with a $0$ in which case it is called a terminating detour.
\end{definition}

If $w_L \ge 1$, according to the previous lemma, either $b_0$ or $b_1$ is equal to $1$. Starting from that first $1$, the rest of the trace of $E$ is a concatenation of detours, which are all terminating except maybe the last detour:
\begin{lemma}
If $w_L \ge 1$, there exist $c \ge 1$ and $c$ integers $v_1 < ... < v_c \le w_L$ such that:\\
- $v_1 = 0$  or  $1$\\
- for all i in $[\![0;c-1]\!]$, $(b_{v_i},...,b_{v_{i+1}-1})$ is a terminating detour\\
- $(b_{v_c},..., b_{w_L})$ is a detour
\end{lemma}

As a convention, let $v_{c+1} = w_L + 1$, so that $(b_{v_i},...,b_{v_{i+1}-1})$ is also the expression of a detour if $i = c$.\\

The number $c$ of detours is less or equal to $\lfloor \frac{w_L}{2}+1 \rfloor$. Indeed for each $i \in [\![1;c-1]\!]$, the $i$th detour contains at least a $1$ and a $0$, whereas the last detour contains at least a $1$; this implies that $w_L+1 \ge 2 (c-1) + 1$, in other words $c \le w_L/2 + 1$.

\subsection{The upper bounds}

We are still with an $[\![n,k,s,m]\!]$ recursive morphism ${\cal{C}}$, and we note $\eta$ its speed. Let $a_N(w,\le d)$ (for $d \in \mathbb{R}$) and $a_N(w,d)$ (for $d \in \mathbb{N}$) be the numbers of sequences $(M,L_1,...,L_N) \in {\cal{P}}^m * {({\cal{P}}^k)}^N$ of weight $w$, which are part of an undetected input sequence:
\[E = (M,L_1,S_1,...,L_N,S_N) \in {\cal{P}}^m * {({\cal{P}}^k * {\cal{Z}}^s)}^N\]
such that, respectively, $|C_N(E)| \le d$ and $|C_N(E)| = d$.
The above characterisation of detours enables to prove the result (the proof is given in appendix):

\begin{theorem}\label{Inner}
\[a_N(w,d) \le a_N(w,\le d) \le 2^m 2^w (|{\cal{P}}|-1)^w \binom{k N + 1}{\lfloor \frac{w}{2} \rfloor + 1} \binom{\eta k(w+d)+1}{\lceil \frac{w}{2} \rceil}\]
\end{theorem}

Notice that $a_N(d, w)$ in the case where ${\cal{C}}$ is the truncated decoder $\bar{C}$ corresponds to the number of possible sequences $(P_1,...,P_N,M')$ of weight $d$ such that the weight of $(M,L_1,...,L_N) = \bar{C}_N (P_1,...,P_N,M')$ has weight $w$. This corresponds exactly to the value of $a_N(w,d)$ in the case where ${\cal{C}}$ is the encoder $C$. This is what will enable us to obtain the two upper bounds on $a_N(w,d)$ concerning the encoder $C$.
By using the following bound on binomials, where $v \le u$:
\[\binom{u}{v} \le \left(\frac{u.e}{v}\right)^v\]
the two upper bounds can be written in the form of this corollary:

\begin{corollary}
Let $C$ be an encoder, and let $a_N(w,\le d)$ and $a_N(w, d)$ be the number of possible sequences $(M,L_1,...,L_N)$ of weight $w$, which are part of an input sequence $E = (M,L_1,S_1,...,L_N,S_N)$ where all the $S_i$ belong to ${\cal{Z}}^s$, and such that $C_N(E)$ has weight respectively less or equal to $d$, and equal to $d$. Then:\\
{\bf Bound 1I:} If $C$ is recursive:
\[ a_N(w,\le d) \le O(1)^w \frac{N^{\frac{w}{2}} (w+d)^{\frac{w}{2}}}{w^w}\]
{\bf Bound 2I:} If $\bar{C}$ is totally recursive:
\[ a_N(w, d) \le O(1)^d \frac{N^{\frac{d}{2}} (w+d)^{\frac{d}{2}}}{d^d}\]
\end{corollary}

\section{The outer encoder}

The outer encoder consists in an $[\![n,k]\!]$ encoder $C$ repeated blockwisely $N$ times ($C$ and $N$ are not the same as the inner encoder).

\begin{definition}
Let $C$ be a $[\![n,k]\!]$ encoder and let $N \in \mathbb{N}^*$. The blockwise encoder $C^{\otimes N}$ is an $[\![Nn,Nk]\!]$ encoder defined by:
\[C^{\otimes N}(L_1,...,L_N,S_1,...,S_N) = C(L_1,S_1).\ ...\ .C(L_N,S_N)\]
where for all $i \le N$, $L_i \in {\cal{P}}^k$ and $S_i \in {\cal{P}}^{n-k}$.\\
\end{definition}

\begin{lemma}{\bf Bound 1E:}
Let $C$ be an encoder of distance $d_c$ and degenerate distance $d_q \ge 2$. For all $d > 0$ and $N > 0$, let $a^{\otimes N}(d)$ be the number of input sequences $E \in {\cal{P}}^{Nk} \times {\cal{Z}}^{N(n-k)}$ such that $|E|_L > 0$ and $|C^{\otimes N}(E)|=d$. Then:
\begin{align*}
a^{\otimes N}(d) & \le O(1)^d \left(\frac{N}{d}\right)^{\frac{d-d_c}{d_q}+1} & \mathrm{if}\ d \ge d_c \\
& = 0 & \mathrm{if} d < d_c
\end{align*}
\end{lemma}

\begin{proof}
Of course if $d > N$ then the bound is true since $a^{\otimes N}(d) = 0$, so we will consider that $d \le N$. Let $E = (L_1,...,L_N,S_1,...,S_N)$ where $|E|_L >0$, $(S_1,...,S_N) \in {({\cal{Z}}^{n-k})}^N$, and such that $|C^{\otimes N}(E)|=d$. Let $j > 0$ be the number of couples $(L_i,S_i)$ different from $(I,I)$. In other words all these $j$ couples are undetected errors for $C$, thus we have $|C(L_i,S_i)| \ge d_q$. Also, since $|E|_L > 0$, at least one of these couples is such that $L_i \ne I$. This couple is a harmful error for $C$ which implies $|C(L_i,S_i)| \ge d_c$. This first shows that necessarily $d\ge d_c$, and:
\[ d \ge (j-1)d_q + d_c\]
and thus:
\[j \le \lfloor\frac{d-d_c}{d_q}\rfloor+1\]
For each possible value for $j$, there are $\binom{N}{j}$ ways to chose the positions of the non zero sequences $(L_i,S_i)$. Each of these sequences can take less than $|{\cal{P}}|^n$ values whereas the remaining sequences are all fixed to $(I,I)$. This leads to the bound:
\[ a^{\otimes N}(d) \le \displaystyle\sum_{j=1}^{\lfloor\frac{d-d_c}{d_q}\rfloor+1} |{\cal{P}}|^{nj} \binom{N}{j}\]
Since $\lfloor\frac{d-d_c}{d_q}\rfloor+1 \le \frac{d}{2} \le \frac{N}{2}$, each term of the sum can be majored by the term where $j = \lfloor\frac{d-d_c}{d_q}\rfloor+1$. Thus:
\[ a^{\otimes N}(d) \le (\lfloor\frac{d-d_c}{d_q}\rfloor+1) |{\cal{P}}|^{n (\lfloor\frac{d-d_c}{d_q}\rfloor+1)} \binom{N}{\lfloor\frac{d-d_c}{d_q}\rfloor+1}\]
Finally the lemma is proved using the bound on binomials, where $v \le u$:
\[\binom{u}{v} \le {\left(\frac{u.e}{v}\right)}^v\]
\end{proof}

\begin{lemma}{\bf Bound 2E:}
Suppose that $|{\cal{P}}| > 2$. Let $C$ be an encoder of degenerate distance $d_q \ge 2$. There exists a constant $c \in ]0,1[$ such that for all $d$ and $N$:
\[a^{\otimes N}(d) \le c^d (|{\cal{P}}| - 1)^d \binom{N n}{d}\]
\end{lemma}

\begin{proof}
Let $X \in {\cal{P}} \backslash {\cal{Z}}$ (a non empty set since $|{\cal{P}}| > 2$), and let $Y = X*Z^{-1}$. For all $i \in [\![0,n]\!]$, let:
\[{\cal{C}}_i = \{ E \in C({\cal{P}}^k * {\cal{Z}}^{n-k}), |E| = i\} \]
And let:
\[c_i = \frac{|{\cal{C}}_i|}{(|{\cal{P}}| - 1)^i\binom{n}{i}}\]

In particular, $c_0 = 1$. Let us show that for $i >0$, $c_i < 1$, or in other words, that $C({\cal{P}}^k * {\cal{Z}}^{n-k})$ does not contain all the sequences of weight $i$. Of course, $c_1 = 0$ since $d_q \ge 2$. Let us consider the case $i > 1$. If $C({\cal{P}}^k * {\cal{Z}}^{n-k})$ contains the two sequences of weight $i$:
\begin{align*}
E_1 = & X^i.I^{n-i} \mathrm{, and}\\
E_2 = & Z.X^{i-1}.I^{n-i}
\end{align*}
then since it is a subgroup of ${\cal{P}}^n$, it contains the sequence of weight $1$:
\[E_1 * E_2^{-1} = Y.I^{n-1}\]
This contradicts the fact that $d_q \ge 2$. Thus, $E_1$ and $E_2$ are not both in ${\cal{C}}_i$, and $c_i < 1$. Let $c = \max\{c_i^{1/i}, 1 \le i \le n\}$, so that $c \in ]0,1[$ and for all $i \in [\![0,n]\!]$, $c_i \le c^i$. Now consider an input sequence $E \in {\cal{P}}^{Nk} \times {\cal{Z}}^{Nn}$, where $|E|_L > 0$ and $|C^{\otimes N}(E)| = d$. The output $C^{\otimes N}(E)$ is the concatenation of $N$ sequences in $C({\cal{P}}^k * {\cal{Z}}^{n-k})$ such that the sum of their weights is equal to $d$. Thus:
\[a^{\otimes N}(d) \le \sum_{\substack{ (d_1,...,d_N) \\ \sum d_i = d \\ 0 \le d_i \le n }} \displaystyle\prod_{i=1}^{N} |{\cal{C}}_{d_i}| \le \sum_{\substack{ (d_1,...,d_N) \\ \sum d_i = d \\ 0 \le d_i \le n }} \displaystyle\prod_{i=1}^{N} c^{d_i} (|{\cal{P}}| - 1)^{d_i}\binom{n}{d_i}\]
The product of the terms $c^{d_i}$ is always equal to $c^d$. The product of the remaining terms, summed over $(d_1,...,d_N)$, counts exactly the number of sequences of size $N n$ and weight $d$. This proves the lemma.

\end{proof}

\section{Upper bounds on the distance of a turbo-encoder}

\subsection{Formal definition of a turbo-encoder}

A turbo-encoder $T_N$ of size $N$ is a concatenation of three operations. It requires a $[\![n_{out},k_{out}]\!]$ encoder $C_{out}$, an interleaver $\Pi$, and an $[\![n_{in},k_{in},m_{in}]\!]$ encoder $C_{in}$. An interleaver is the following operation:
\begin{definition}
An interleaver $\Pi$ of size $N$ is an automorphism of ${\cal{P}}^N$ composed of two operations:
\begin{itemize}
\item First, a permutation $\pi$ of the $N$ positions of the letters of the sequence
\item Second, a sequence $(\pi_1,...,\pi_N)$ of $N$ automorphisms of ${\cal{P}}$, applied at each one of the letters
\end{itemize}
Thus an interleaver $\Pi$ transforms a sequence $E = E_1.\ ...\ .E_N$ into: 
\[ \pi_1(E_{\pi(1)}).\ ...\ .\pi_N(E_{\pi(N)})\]

The set of interleavers of size $N$ is noted $\mathbb{P}_N$.
\end{definition}

\begin{definition}
Consider:
\begin{itemize}
 \item $C_{out}$ an $[\![n_{out},k_{out}]\!]$ encoder and $C_{in}$ an $[\![n,k,m]\!]$ encoder
 \item $N = N_{out} \in \mathbb{N^*}$ and $N_{in} \in \mathbb{N^*}$ such that $N_{out} n_{out} = N_{in} k_{in} + m_{in}$
 \item $\Pi$ an interleaver of size $N_{out} n_{out}$.
\end{itemize}
The turbo-encoder $T_N$ based on $C_{out}$, $\Pi$, and $C_{in}$, is an $[\![N_{in} n_{in} + m_{in}, N_{out} k_{out}]\!]$ encoder which does the following transformation. Consider a sequence $E \in {\cal{P}}^{N_{out} k_{out}} \times {\cal{P}}^{N_{in} (n_{in} - k_{in})}$ written in the form:
\[E = (L_1,...,L_{N_{out}},S_1,...,S_{N_{out}},S'_1,...,S'_{N_{in}})\]
Then $T_N(E)$ is obtained by these three steps:
\begin{itemize}
 \item First, apply $C_{out}^{\otimes N}$ at the first $2 N$ errors of $E$:
\[E' = C_{out}^{\otimes N}(L_1,...,L_N,S_1,...,S_N)\]
\item Then, apply $\Pi$ at $E'$:
\[E'_{perm} = \Pi(E')\]
Since $N_{out} n_{out} = N_{in} k_{in} + m_{in}$, $E'_{perm}$ can be written in the form:
\[E'_{perm} = (M',L'_1,...,L'_{N_{in}})\]
where $M'$ is of size $m_{in}$ and all the other errors $L_i$ are of size $k_{in}$. Note that the first $m_{in}$ letters of $E'_{perm}$ are specialized into memory letters for the next step.
\item Finally, insert $(S'_1,...,S'_{N_{in}})$ into $E'_{perm}$ by putting each $S'_i$ after $L'_i$, and apply ${C_{in}}_{N_{in}}$:
\[T_N(E) = {C_{in}}_{N_{in}}(M',L'_1,S'_1,...,L'_{N_{in}},S'_{N_{in}})\]
Note that even the memory part of this last output counts as part of the physical output with respect to the global protocol.
\end{itemize}
A random turbo-encoder based on $C_{out}$ and $C_{in}$ is a turbo-encoder based on $C_{out}$, an interleaver $\Pi$ chosen randomly with a uniform distribution over $\mathbb{P}_{N n_{out}}$, and $C_{in}$.
$N$ is called the length of the turbo-encoder, and an integer $N$ is said to be eligible for $C_{out}$ and $C_{in}$ if there exists an integer $N_{in}$ such that $N n_{out} = N_{in} k_{in} + m_{in}$.
\end{definition}

If we go back to describing the real protocol lying behind a turbo-encoder, a state of information of size $N_{out} k_{out}$ is encoded into a state of size $N_{in} n_{in} + m_{in}$, by first encoding it into a state of size $N_{out} n_{out} = N_{in} k_{in} + m_{in}$, then interleaving the positions of the obtained state, then encoding it again into a state of size $N_{in} n_{in} + m_{in}$. At each of the first and the last steps of the encoding, ancillary positions are added before the encoding is done. This is what corresponds, in the formal protocol, to the insertion of the stabilizer errors $S_i$ and $S'_i$.

From now on, for a given turbo-encoder $T_N$ and a given input sequence $E$ of the turbo-encoder, we will systematically write $E'$ and $E'_{perm}$ to refer to the intermediate sequences obtained after applying $C_{out}^{\otimes N}$ and after applying $\Pi$.

\subsection{Sketch of the counting argument}

We will sew the results obtained on the weight distributions for the inner encoder $a_{N_{in}}(w,d)$, $a_{N_{in}}(w,\le d)$, and the outer encoder $a^{\otimes N}(d)$, $a^{\otimes N}(\le d)$, to obtain the desired upper bounds on the distance of the turbo-encoder. For this purpose, the tools will be the two following lemmas. Let $C_{out}$ be a $[\![n_{out},k_{out}]\!]$ encoder and $C_{in}$ be an $[\![n_{in},k_{in},m_{in}]\!]$ encoder. Let $T_N$ be a random turbo-encoder based on $C_{out}$ and $C_{in}$ where $N$ is an eligible integer for $C_{out}$ and $C_{in}$. Consider the probabilities related to the following events:
\begin{itemize}
 \item $p_N(d)$ and $p_N(\le d)$: there exists a harmful input sequence $E$ such that, respectively, $|T_N(E)| = d$ and $|T_N(E)| \le d$. Note that $p_N(\le d)$ is the probability that $d_c(T_N) \le d$.
 \item $p_N(w,d)$ and $p_N(w,\le d)$: there exists a harmful input sequence $E$ such that $|E'|=w$ and, respectively, $|T_N(E)| = d$ and $|T_N(E)| \le d$.
\end{itemize}

Since the interleaver $\Pi$ is chosen at random, these probabilities are only functions of $C_{out}$, $C_{in}$ and $N$. These probabilities are defined for all integers $w$ and $d$, and $p_N(\le d)$ and $p_N(w,\le d)$ are also defined when $d$ is real.

\begin{lemma}
\begin{align*}
\forall (w,d) \in \mathbb{N} \times \mathbb{N}, & \ p_N(w,d) \le \frac{a^{\otimes N}(w)\ a_{N_{in}}(w,d)}{{(|{\cal{P}}| - 1)}^w \binom{N n_{out}}{w}}\\
\forall (w,d) \in \mathbb{N} \times \mathbb{R}, & \ p_N(w,\le d) \le \frac{a^{\otimes N}(w)\ a_{N_{in}}(w,\le d)}{{(|{\cal{P}}| - 1)}^w \binom{N n_{out}}{w}}
\end{align*}
\end{lemma}

\begin{proof}
Let us prove the first inequality, the second one can be obtained with a similar reasoning. $p_N(w,d)$ is the probability that there exists an input sequence $E$ for the turbo-encoder $T_N$ such that $|E|_L > 0$, $|E'| = w$ and $|T_N(E)| = d$. In order to obtain this, a necessary condition is that the first $2 N$ errors of $E$ constitute one of the $a^{\otimes N}(w)$ harmful input sequences of $C_{out}^{\otimes N}$ such that the output (by $C_{out}^{\otimes N}$) has weight $w$. Take such a sequence of $2 N$ errors. $E'$ is now uniquely defined. Since $E'_{perm}$ is the image of $E'$ under the action of a random interleaver $\Pi$, it is uniformly distributed over the set of sequences of size $N n_{out}$ and weight $w$. Having $|T_N(E)| = d$ implies that $E'_{perm}$ is one of the $a_{N_{in}}(w,d)$ sequences of weight $w$ which are part of an undetected input sequence of ${C_{in}}_{N_{in}}$ such that the output by ${C_{in}}_{N_{in}}$ has weight $d$. This has probability:
\begin{eqnarray*}
\frac{a_{N_{in}}(w,d)}{{(|{\cal{P}}| - 1)}^w \binom{N n_{out}}{w}}
\end{eqnarray*}
The inequality of the lemma is proved using the union bound, by summing this probability over all the $a^{\otimes N}(w)$ possible values of the first $2 N$ errors.
\end{proof}

The second lemma is straightforward by using the union bound:

\begin{lemma}
For any real number $D$, and for all $x \in ]0,n_{in}[$:
\begin{align*}
p_N(\le D)
& \le \displaystyle\sum_{w=0}^{\lfloor D \rfloor} p_N(w,\le d) + \displaystyle\sum_{d=0}^{\lfloor D \rfloor}\displaystyle\sum_{w=\lfloor D \rfloor + 1}^{\lfloor x N \rfloor - 1} p_N(w,d) + \displaystyle\sum_{d = 0}^{\lfloor D \rfloor}\displaystyle\sum_{w= \lfloor x N \rfloor}^{N n_{out}} p_N(w,d)
\end{align*}
We call these three terms first, second and third partial sum.
\end{lemma}

\begin{proof}
\begin{align*}
p_N(\le D) & \le \displaystyle\sum_{w=0}^{N n_{out}} p_N(w,\le D)\\
& \le \displaystyle\sum_{w=0}^{\lfloor D \rfloor} p_N(w,\le D) + \displaystyle\sum_{w= \lfloor D \rfloor + 1}^{N n_{out}} p_N(w,\le D)\\
& \le \displaystyle\sum_{w=0}^{\lfloor D \rfloor} p_N(w,\le D) + \displaystyle\sum_{d=0}^{\lfloor D \rfloor}\displaystyle\sum_{w=\lfloor D \rfloor + 1}^{N n_{out}} p_N(w,d)\\
& \le \displaystyle\sum_{w=0}^{\lfloor D \rfloor} p_N(w,\le D) + \displaystyle\sum_{d=0}^{\lfloor D \rfloor}\displaystyle\sum_{w=\lfloor D \rfloor + 1}^{\lfloor x N \rfloor - 1} p_N(w,d) + \displaystyle\sum_{d = 0}^{\lfloor D \rfloor}\displaystyle\sum_{w= \lfloor x N \rfloor}^{N n_{out}} p_N(w,d)
\end{align*}
\end{proof}

\subsection{The polynomial bounds}

We will start by proving case $1$ of Theorem \ref{Th1}, then Theorem \ref{Th2} since these results are close to each other, then we will prove case $2$ of Theorem \ref{Th1}. The first two results are obtained using the three partial sums when $D = N^{\alpha}$. With the three following lemmas proved in Appendix $B.1$, we show successively that under the corresponding conditions, each of these three partial sums tends to $0$ as $N \rightarrow \infty$. We write $d_c$ and $d_q$ respectively for the distance and the degenerate distance of $C_{out}$.

\begin{lemma}\label{First poly} {\bf First partial sum, poly case}
If $d_q \ge 2$ and $C_{in}$ is recursive, for all $\alpha < \frac{d_q-2}{d_q}$:
\[\lim_{N \rightarrow \infty} \displaystyle\sum_{w=0}^{\lfloor N^{\alpha} \rfloor} p_N(w,\le N^{\alpha}) = 0\]
\end{lemma}

\begin{lemma}\label{Second poly} {\bf Second partial sum, poly case}
If $d_q > 2$ and $\bar{C_{in}}$ is totally recursive, then for all $\alpha < \frac{d_q-2}{d_q}$, there exists $x > 0$ such that :
\[\lim_{N \rightarrow \infty} \displaystyle\sum_{d=0}^{\lfloor N^{\alpha} \rfloor}\displaystyle\sum_{w= \lfloor N^{\alpha} \rfloor + 1}^{\lfloor x N \rfloor - 1} p_N(w,d) = 0\]
\end{lemma}

\begin{lemma}\label{Third poly} {\bf Third partial sum, poly case}
If $|{\cal{P}}| > 2$, $d_q \ge 2$ and $\bar{C_{in}}$ is totally recursive, then for all $x > 0$:
\[\lim_{N \rightarrow \infty} \displaystyle\sum_{d = 0}^{\lfloor N^{\alpha} \rfloor}\displaystyle\sum_{w= \lfloor x N \rfloor}^{N n_{out}} p_N(w,d) = 0\]
\end{lemma}

The proof of Theorem \ref{Th1} case $1$ is now the following.
\begin{proof}
The hypothesis of the theorem are that $|{\cal{P}}| > 2$, $d_q > 2$, $C_{in}$ is recursive and $\bar{C_{in}}$ is totally recursive. The conditions are met so that the results from the lemmas above apply. For all $\alpha < \frac{d_q-2}{d_q}$, let $x$ such that the second partial sum goes to $0$. Let us write:
\begin{align*}
p_N(\le N^{\alpha})
& \le \displaystyle\sum_{w=0}^{\lfloor N^{\alpha} \rfloor} p_N(w,\le d) + \displaystyle\sum_{d=0}^{\lfloor N^{\alpha} \rfloor}\displaystyle\sum_{w=\lfloor N^{\alpha} \rfloor + 1}^{\lfloor x N \rfloor - 1} p_N(w,d) + \displaystyle\sum_{d = 0}^{\lfloor N^{\alpha} \rfloor}\displaystyle\sum_{w= \lfloor x N \rfloor}^{N n_{out}} p_N(w,d)
\end{align*}
Since the three partial sums go to $0$, this proves that $p_N(\le N^{\alpha})$ goes to $0$ as $N \rightarrow \infty$.
\end{proof}

The proof of Theorem \ref{Th2} is the following:
\begin{proof}
The hypothesis of the theorem are that $d_q > 2$, and $C_{in}$ is recursive and systematic. Being systematic implies that if $w > d$, $p_N(w,d) = 0$. Thus the second and the third partial sums are equal to $0$. The conditions of Lemma \ref{First poly} are met and thus again, $p_N(\le N^{\alpha})$ goes to $0$ as $N \rightarrow \infty$.
\end{proof}

\subsection{The sublogarithmic bound}

Now let us focus on Theorem \ref{Th1} case $2$. For this purpose we consider the three partial sums when $D = \log \log N$. The following proposition justifies the sublogarithmic expression of the bound. This property is needed in order to upper bound the first partial sum.
\begin{proposition}\label{llog^llog}:
 Let us write, for $N > 1$, $\mathrm{llog}\ N = \log N/\log(\log N)$. For all $t > 0$, if $N$ is big enough( precisely, if $N \ge e^{e^t}$):
\[ {(t \ \mathrm{llog}\ N)}^{t \ \mathrm{llog}\ N} \le N^t \]
\end{proposition}

\begin{proof}
If $N \ge e^{e^t}$, $\mathrm{llog}(N)$ is positive and its logarithm is well defined, and the logarithm of the left hand side of the inequality is equal to:
\begin{align*}
(t \ \mathrm{llog}\ N) (\log t + \log (\mathrm{llog}\ N)) & = t \frac{\log N}{\log(\log N)} (\log t + \log(\log N) - \log(\log(\log N)))\\
& = t \log N  + t \frac{\log N}{\log(\log N)} \log \left( \frac{t}{\log \log N} \right) \\
& \le t \log N
\end{align*}
\end{proof}
Actually one can prove that the solution of the equation $x^x = N^t$ is equivalent to $t \log N/\log(\log N)$, and thus the bound proposed is the best asymptotic distance which can be proved with the arguments presented in this paper. The three following lemmas proved in Appendix $B.2$ show under the corresponding conditions that each of the three partial bounds, when $D = \log \log N$, tend to $0$ as $N \rightarrow \infty$.

\begin{lemma}\label{First sublog} {\bf First partial sum, sub-log case}
If $d_c > d_q = 2$ and $C_{in}$ is recursive, for all $\alpha < d_c-2$:
\begin{align*}
\lim_{N \rightarrow \infty} \displaystyle\sum_{w=0}^{\lfloor \alpha\ \mathrm{llog}\ N \rfloor} p_N(w, \le \alpha\ \mathrm{llog}\ N) = 0
\end{align*}
\end{lemma}

\begin{lemma}\label{Second sublog} {\bf Second partial sum, sub-log case}
If $d_c > d_q = 2$ and $\bar{C_{in}}$ is totally recursive, for all $\alpha < d_c-2$, there exists $x$ such that:
\[\lim_{N \rightarrow \infty} \displaystyle\sum_{d=0}^{\lfloor \alpha\ \mathrm{llog}\ N \rfloor}\displaystyle\sum_{w= \lfloor \alpha\ \mathrm{llog}\ N \rfloor + 1}^{\lfloor x N \rfloor - 1} p_N(w,d)=0\]
\end{lemma}

\begin{lemma}\label{Third sublog} {\bf Third partial sum, sub-log case}
If $|{\cal{P}}| > 2$, $d_q \ge 2$ and $\bar{C_{in}}$ is totally recursive, then for all $x > 0$:
\[\lim_{N \rightarrow \infty} \displaystyle\sum_{d=0}^{\lfloor \alpha\ \mathrm{llog}\ N\rfloor}\displaystyle\sum_{w=\lfloor x N \rfloor}^{N n_{out}} p_N(w,d) = 0\]
\end{lemma}

The proof of Theorem \ref{Th1} case $2$ is now the following.
\begin{proof}
The hypothesis of the theorem are that $|{\cal{P}}| > 2$, $d_c > d_q = 2$, $C_{in}$ is recursive and $\bar{C_{in}}$ is totally recursive. The conditions are met so that the results from the lemmas above apply. For all $\alpha < d_c-2$, let $x$ such that the second partial sum goes to $0$. Let us write:
\begin{align*}
p_N(\alpha\ \mathrm{llog}\ N) \le \displaystyle\sum_{w=0}^{\lfloor \alpha\ \mathrm{llog}\ N \rfloor} p_N(w, \le \alpha\ \mathrm{llog}\ N) + \displaystyle\sum_{d=0}^{\lfloor \alpha\ \mathrm{llog}\ N \rfloor}\displaystyle\sum_{w=\lfloor \alpha\ \mathrm{llog}\ N \rfloor + 1}^{\lfloor x N\rfloor - 1} p_N(w,d) + \displaystyle\sum_{d=0}^{\lfloor \alpha\ \mathrm{llog}\ N \rfloor}\displaystyle\sum_{w=\lfloor x N\rfloor}^{N n_{out}} p_N(w,d)
\end{align*}
Since the three partial sums go to $0$, this proves that $p_N(\alpha\ \mathrm{llog}\ N)$ goes to $0$ as $N \rightarrow \infty$.
\end{proof}

\section*{Appendix A: Bounds for the inner convolutional encoder}

We give here the proof of Theorem \ref{Inner}. Consider a sequence $(M,L_1,...,L_N)$ of weight $w$, which is part of $E = (M,L_1,S_1,...,L_N,S_N)$ such that $|{\cal{C}}_N(E)| \le d$, and let $(P_1,...,P_N,M') = {\cal{C}}_N(E)$. To begin with, suppose that the information part is of given weight $|(L_1,...,L_N)| = w_L$, and suppose that the trace of $E$ is made of $c$ detours. We remind the following notations. $1 \le p_1 < ... < p_{w_L} \le {k N}$ are the positions of the $w_L$ non identity letters of $(L_1,...,L_N)$, and for each $i \in [\![1;w_L]\!]$, $N_i = \lceil p_i / k \rceil$ is the index such that the $i$th non identity letter is part of the error $L_{N_i}$. The $i$th truncature of $E$ is the sequence:
\[E_{\backslash i} = (M,L_1,S_1,...,L_{N_{i-1}}, S_{N_{i-1}}, L'_{N_i}, S_{N_i})\]
where $(L_1,...,L_{N_{i-1}},L'_{N_i})$ is obtained by replacing all the letters in $(L_1,...,L_{N_i})$ after the position $p_i$ by an $I$. Also, $p_0 = 0$, $N_0 = 0$ and $E_{\backslash 0} = M$. Moreover:
\[M_i = \mu_{N_i} (E_{\backslash i})\]
and for $i = 0$, $M_0 = M$. The trace of $E$ is the sequence $(b_0,...,b_{w_L})$ such that for $i \in [\![1;w_L]\!]$, $M_i \in \mathbb{M}_{b_i}$. For $i \in [\![1;c]\!]$, $v_i$ is the starting point of the $i$th detour, and $v_{c+1} = w_L + 1$.

We will start by confining the space where lie the positions of the non identity letters of $(L_1,...,L_N)$, by cutting them into packets corresponding to each detour. For each $i \in [\![1;c]\!]$, let $\delta p^{(i)} = p_{v_{i+1}-1}-p_{v_i}$. The sum of the sizes of these intervals is upper bounded by:

\begin{lemma}
If $w_L \ge 1$:
\[\displaystyle\sum_{i=1}^{c} \delta p^{(i)} \leq min(k N, \eta k (w_L + d))\]
where $\delta p^{(i)} = p_{v_{i+1}-1}-p_{v_i}$.
\end{lemma}

The main idea is to prove that the difference of weight between the outputs at the steps $N_{v_i}$ and $N_{v_{i+1}}$ of the convolutional operation is proportional to $\delta p^{(i)}$. The main argument is that for $j$ from $v_i$ to $v_{i+1}-1$, the memory errors $M_j$ all belong to $\mathbb{M}_1$, which by the recursiveness of ${\cal{C}}$ yields an output weight proportional to the input size during that time. We first apply this idea between the steps $N_j$ and $N_{j+1}-1$ of the encoding for a given $j \in [\![v_0;w_L]\!]$:

\begin{sublemma}
Let $j \in [\![v_0;w_L]\!]$. \[(|P_{N_j+1},...,P_{N_{j+1}-1})| \ge  \frac{N_{j+1} - N_j}{\eta} - 1 \]
\end{sublemma}

\begin{proof}
If $N_{j+1} = N_j$ or $N_j+1$, this is true since the sequence $(P_{N_j+1},...,P_{N_{j+1}-1})$ is empty, and its weight is non negative.\\

Now suppose that $N_{j+1} > N_j + 1$, and let us first consider the case $j > 0$. Notice that $L'_{N_j} = L_{N_j}$, because the non identity information letter in position $p_{j+1}$ belongs to the error $L_{N_{j+1}}$, which comes strictly after the error $L_{N_j}$. We can write the following concatenation:
\[E_{\backslash j+1} = E_{\backslash j}.(I,S_{N_j + 1},...,I,S_{N_{j+1}-1},L'_{N_{j+1}},S_{N_{j+1}})\]
The physical output $(P_1,...,P_{N_{j+1}-1})$ is produced by applying the encoder at the first $2 N_{j+1}-1$ errors of $E_{\backslash j+1}$. By concatenation, this output can be written:
\begin{eqnarray*}
(P_1,...,P_{N_{j+1}-1}) & = & \pi_{N_{j+1}-1}(E_{\backslash j}.(I,S_{N_j + 1},...,I,S_{N_{j+1}-1}))\\
& = & \pi_{N_j}(E_{\backslash j}).\pi_{N_{j+1} - N_j - 1}(M_j,I,S_{N_j + 1},...,I,S_{N_{j+1}-1})\\
& = & (P_1,...,P_{N_j}).\pi_{N_{j+1} - N_j - 1}(M_j,I,S_{N_j + 1},...,I,S_{N_{j+1}-1})
\end{eqnarray*}
This shows that:
\[(P_{N_j+1},...,P_{N_{j+1}-1}) = \pi_{N_{j+1} - N_j - 1}(M_j,I,S_{N_j + 1},...,I,S_{N_{j+1}-1})\]
In the case where $j = 0$, this equality also holds and has the form:
\[(P_1,...,P_{N_1-1}) = \pi_{N_1 - 1}(M_0,I,S_1,...,I,S_{N_1-1})\]
Since $M_j \in \mathbb{M}_1$ and ${\cal{C}}$ is recursive:
\[|(P_{N_j+1},...,P_{N_{j+1}-1})| \ge \lfloor \frac{N_{j+1} - N_j - 1}{\eta} \rfloor\]
And this implies the desired inequality.
\end{proof}

Now by simply summing over $j$ we obtain the proof of the lemma.

\begin{proof}
Let $i \in [\![1;c]\!]$. The sum of the inequality in the previous lemma over all the values of $j \in [\![v_i;v_{i+1}-2]\!]$ gives the result:
\[|(P_{N_{v_i}},...,P_{N_{v_{i+1}-1}})| \ge \frac{N_{v_{i+1}-1} - N_{v_i}}{\eta} - (v_{i+1} - v_i - 1)\]
Using $N_{v_i} \le p_{v_i} / k + 1$ and $N_{v_{i+1}-1} \ge p_{v_{i+1}-1} / k$ we get:
\[|(P_{N_{v_i}},...,P_{N_{v_{i+1}-1}})| \ge \frac{p_{v_{i+1}-1} - p_{v_i}}{\eta k} - \frac{1}{\eta} - (v_{i+1} - v_i - 1)\]
By summing again this inequality over $i \in [\![1; c]\!]$ we get that:
\[d \ge \displaystyle\sum_{i=1}^{c} \frac{\delta p^{(i)}}{\eta k} - w_L\]
This implies the first bound:
\[\displaystyle\sum_{i=1}^{c} \delta p^{(i)} \leq \eta k (w_L + d)\]
The second bound is obvious:
\[\displaystyle\sum_{i=1}^{c} \delta p^{(i)} = \displaystyle\sum_{i=1}^{c} (p_{v_{i+1}-1}-p_{v_i}) \le p_{w_L} - p_{v_c} + \displaystyle\sum_{i=1}^{c-1} (p_{v_{i+1}} - p_{v_i}) \le k N\]
\end{proof}

This upper bound is useful in that it says that all the non identity letters of $(L_1,...,L_N)$ starting from the letter in position $p_{w_1}$ (ie, from the first non identity letter since ${w_1} = 0$ or $1$) are confined in a bounded region of space. Now we can prove the following bound:

\begin{lemma}
Suppose $w_L \ge 1$, and suppose that $c$ is fixed. The number of possible values for the sequence $(L_1,...,L_N)$ is upper bounded by:
\[ (|{\cal{P}}|-1)^{w_L} \binom{w_L}{c-1} \binom{k N}{c} \binom{min \left(k N, \eta k (w_L + d) \right) + 1}{w_L -c+1}\]
\end{lemma}

\begin{proof}
There are two possible values for the starting point $w_1$ of the first detour, either $0$ or $1$. In both cases, the number of possible starting points $w_i$, $2 \leq i \leq c$ of the $c-1$ remaining detours is upper bounded by $\binom{w_L}{c-1}$, since the $c-1$ remaining values are in the interval $[\![1,v_c]\!]$. For each choice of these starting points, the number of possible values for the positions $(p_{v_i})_{1 \leq i \leq c}$, is at most $\binom{k N}{c}$.\\

Now, consider that the sequences $(v_i)_{1 \leq i \leq c}$ and $(p_{v_i})_{1 \leq i \leq c}$ are fixed. The number of positions of the non identity letters in $(L_1,...,L_N)$ which are still unfixed is equal to $w_L-c$ if $v_1 = 0$ and $w_L-c+1$ if $v_1 = 1$. This sequence of remaining positions can be written:
\[(p_j)_{v_i < j < v_{i+1},\ 1 \leq i \leq c}\]
where as defined previously, $v_{c+1} - 1 = w_L$. In order to upper bound the number of such possible remaining sequences, we will rely on the fact that $\sum_{i=1}^{c} \delta p^{(i)}$, which we know is upper bounded, is intuitively the space in which the remaining positions are confined. Since this sequence is strictly growing, it is equivalent to upper bound the number of possible corresponding sets. Such a corresponding set can be written as:
\[\displaystyle\bigcup_{i=1}^{c}{\cal{I}}_i\]
where:
\[{\cal{I}}_i = \{p_j, v_i < j < v_{i+1}\}\ .\]
In the following, we exhibit a reversible transformation of this set into a set of equal number of elements and included in the interval $[\![1,\sum_{i=1}^{c} \delta p^{(i)}]\!]$. Define $s_1 = 0$ and for $i \in [\![2,c+1]\!]$, let $s_i = \sum_{i'=1}^{i-1} \delta p^{(i')}$. Consider the transformation which transforms the set of remaining positions into :
\[\displaystyle\bigcup_{i=1}^{c}\bar{\cal{I}}_i\]
where:
\[\bar{\cal{I}}_i = \{\bar{p}_j = p_j - p_{v_i} + s_i, v_i < j < v_{i+1} \}\ .\]
This function relies on the knowledge of $v_i$ and $p_{v_i}$. Let us show that this transformation is injective by showing explicitely how to recover the antecedent of a given image set.\\

First, let us show that in the image set, any two elements $\bar{p}_j$ and $\bar{p}_{j'}$ respect the order of their indexes, ie they verify $\bar{p}_j < \bar{p}_{j'}$ if and only if $j < j'$.\\

For each $i \in [\![1,c]\!]$, and for each couple $(\bar{p}_j,\bar{p}_{j'}) \in {\bar{\cal{I}}_i}^2$ where $j \le j'$:
\begin{align*}
\bar{p}_j & \ge 1 + s_i & \mathrm{because}\ & p_j - p_{v_i} \ge 1\\
\bar{p}_{j'} & \le s_{i+1} & \mathrm{because}\ & p_{j'} - p_{v_i} \le \delta p^{(i)}\\
\bar{p}_j & \le \bar{p}_{j'} & \mathrm{because}\ & p_j \le p_{j'}
\end{align*}
If we compare two elements in a same set $\bar{\cal{I}}_i$, the third inequality proves that they respect the order of their indexes, whereas if $\bar{p}_j \in \bar{\cal{I}}_i$ and $\bar{p}_{j'} \in \bar{\cal{I}}_{i'}$ with $i < i'$, we can use the first two inequalities to get:
\[\bar{p}_j \le s_{i+1} \le 1 + s_{i'} \le 1 + \bar{p}_{j'}\]
This implies that, by looking at the position of an element respectively to the others in the image set, we know necessarily the value of its corresponding index $j$. Next, the $c+1$ values of $s_i$ can be recovered recursively as follows. We start with $s_1 = 0$. Then, for each $i \in [\![1, c ]\!]$, either $\bar{\cal{I}}_i$ is empty, in which case $v_{i+1} = v_i+1$ and $s_{i+1} = s_i$, or it is not, in which case $s_{i+1}$ is equal to the last element of $\bar{\cal{I}}_i$: $\bar{p}_{v_{i+1}-1} = \delta p^{(i)} + s_i = s_{i+1}$. Now that all the values of $p_j$ are known, and since the values of $p_{v_i}$ are also known, all the values of $p_j$ can be recovered. Finally, the fact that all the values $\bar{p}_j$ are confined in the set $[\![1, \sum_{i=1}^{c} \delta p^{(i)} ]\!]$ comes from the first two inequalities.\\

This proves respectively that, when $v_1 = 0$ and when $v_1 = 1$, the number of possible sequences of remaining positions $(p_j)_{v_i < j < v_{i+1},\ 1 \leq i \leq c}$ is upper bounded by the number of subsets of $[\![1, min \left(k N, \eta k (w_L + d) \right) ]\!]$ of respective number of elements $w_L - c$ and $w_L -c+1$. This gives two binomials, the sum of which is:
\[\binom{min \left(k N, \eta k(w_L+d) \right)+1}{w_L-c+1}\]
Finally, the proof is completed by noticing that each of the $w_L$ non identity letters takes at most $|{\cal{P}}|-1$ values.
\end{proof}

Now let us release the constraint on the number of detours $c$:
\begin{lemma}
Suppose that $w_L \ge 0$. The number of possible values for the sequence $(L_1,...,L_N)$ is upper bounded by:
\[2^{w_L}(|{\cal{P}}|-1)^{w_L} \binom{k N + 1}{\lfloor \frac{w_L}{2} \rfloor + 1} \binom{min \left(k N, \eta k(w_L+d)\right)+1}{\lceil \frac{w_L}{2} \rceil}\]
\end{lemma}

\begin{proof}
The sequence of $N$ identity errors is the only one such that $w_L = 0$. Thus if $w_L = 0$, the bound is true. Now we suppose that $w_L \ge 1$. For a purpose of lisibility, let ${\cal{B}} = min \left(k N, \eta k (w_L + d) \right) + 1$, and let $c_{max} = \lfloor \frac{w_L}{2} \rfloor + 1$ be the maximal value for the number $c$ of detours. By the previous lemma, he number of possible values for the sequence $(L_1,...,L_N)$ is upper bounded by:
\[(|{\cal{P}}|-1)^{w_L} \displaystyle\sum_{c=1}^{c_{max}} \binom{w_L}{c-1} \displaystyle\max_{1 \le c \le c_{max}} \binom{k N}{c} \binom{{\cal{B}}}{w_L -c+1}\]
Replacing $k N$ by $k N + 1$ in the previous expression keeps the inequality true and makes the analysis easier. Let us maximize over $c$, and under the constraint $c \le c_{max}$, the expression:
\[f(c) = \binom{k N + 1}{c} \binom{{\cal{B}}}{w_L -c+1}\]
For all $c \le c_{max}-1$, the ratio $r(c)$ between two successive values is:
\[r(c) = \frac{f(c+1)}{f(c)} = \frac{k N+1-c}{{\cal{B}} -(w_L-c)}\frac{w_L-c+1}{c+1}\]
When $c \le c_{max}-1$, $c \le w_L/2$, and thus $c \le w_L - c$. Besides, ${\cal{B}} \le k N+1$. Thereby, $r(c) \ge 1$ for all $c \le c_{max}-1$, which proves that the maximum is reached at $c = c_{max}$. When $c$ takes this value, $w_L -c+1$ is equal to $\lceil \frac{w_L}{2} \rceil$; this comes from the fact that $\lfloor \frac{w_L}{2} \rfloor + \lceil \frac{w_L}{2} \rceil = w_L$. Thus:
\[\displaystyle\max_{1 \le c \le c_{max}} \binom{k N}{c} \binom{{\cal{B}}}{w_L -c+1} = \binom{k N + 1}{\lfloor \frac{w_L}{2} \rfloor + 1} \binom{min \left(k N, \eta k(w_L+d)\right)+1}{\lceil \frac{w_L}{2} \rceil}\]
Moreover:
\[\displaystyle\sum_{c=1}^{c_{max}} \binom{w_L}{c-1} \le 2^{w_L}\]
and this proves the lemma.
\end{proof}

Finally comes the proof of Theorem \ref{Inner}:

\begin{proof}
The number of memory errors $M$ of given weight $w_M$ is:
\[(|{\cal{P}}|-1)^{w_M} \binom{m}{w_M}\]
Thus by the previous lemma, the number of possible sequences $(M,L_1,...,L_N)$ verifying the conditions of the theorem and such that $|M| = w_M$ is upper bounded by:
\[(|{\cal{P}}|-1)^{w}\binom{m}{w_M} g(w_L)\]
where $w_L = w-w_M$ and where the function $g$ is given by:
\[g(w_L) = 2^{w_L} \binom{k N + 1}{\lfloor \frac{w_L}{2} \rfloor + 1} \binom{\eta k(w+d)+1}{\lceil \frac{w_L}{2} \rceil}\]
Again, we search for the maximum of $g$ for all values of $w_L$ under the constraint $w_L \le w \le k N$. If $w_L \le w-1$ then these two inequalities:
\[\lfloor \frac{w_L}{2} \rfloor + 1 \le \frac{1}{2}(k N+1)\]
and:
\[\lceil \frac{w_L}{2} \rceil \le  \frac{1}{2}(\eta k(w+d)+1)\]
show that $g(w_L+1) > g(w_L)$. Then, $g$ reaches its maximum at $w_L = w$. The proof is completed by the majoration:
\[\displaystyle\sum_{w_M=0}^{m} \binom{m}{w_M} \le 2^m\]
\end{proof}

\section*{Appendix B: partial sums}

\subsection*{B.1: Polynomial case}

Proof of Lemma \ref{First poly}:
\begin{proof}{\bf First partial sum, poly case}
Let $d = N^{\alpha}$. We use the following property:
\[p_N(w,\le d) \le \frac{a^{\otimes N}(w)\ a_{N_{in}}(w,\le d)}{{(|{\cal{P}}| - 1)}^w \binom{N n_{out}}{w}}\]
with the following bound derived from the bound (1E), valid because $d_q \ge 2$:
\begin{align*}
a^{\otimes N}(w) & \le O(1)^w \left(\frac{N}{w}\right)^{\frac{w}{d_q}} & \mathrm{if}\ w \ge d_c \\
& = 0 & \mathrm{if} w < d_c
\end{align*}
and the bound (1I), valid because $C_{in}$ is recursive:
\[ a_{N_{in}}(w,\le d) \le O(1)^w \frac{N_{in}^{\frac{w}{2}} (w+d)^{\frac{w}{2}}}{w^w}\]
Using the fact that $w \le d$ and $N_{in} = O(N)$, this bound gives :
\[a_{N_{in}}(w,\le d) \le O(1)^w \frac{N^{\frac{w}{2}} d^{\frac{w}{2}}}{w^w}\]
We lower bound the denominator using the following lower bound on binomials:
\[\forall (u,v) \in \mathbb{N}^2 / v \le u, \binom {u}{v} \ge \left(\frac{u}{v}\right)^v\]
which, together with the facts that $|{\cal{P}}| \ge 2$ and $n_{out} \ge 1$, implies :
\[{(|{\cal{P}}| - 1)}^w \binom{N n_{out}}{w} \ge \left(\frac{N}{w}\right)^w\]
As a consequence we obtain, when $w \ge d_c$ :
\begin{align*}
p_N(w,\le d) \le & O(1)^w \left(\frac{N}{w}\right)^{\frac{w}{d_q}} \frac{N^{\frac{w}{2}} d^{\frac{w}{2}}}{w^w} \left(\frac{N}{w}\right)^{-w}\\
\le & O(1)^w N^{\frac{w}{d_q}} N^{\frac{w}{2}} d^{\frac{w}{2}} N^{-w}\\
\le & {\left[O(1) N^{{\frac{1}{d_q}-\frac{1}{2}}} d^{\frac{1}{2}} \right]}^w
\end{align*}
and when $w < d_c$, $p_N(w,\le d) = 0$.
Let us simply replace $d$ by its value $N^{\alpha}$:
\begin{align*}
p_N(w,\le N^{\alpha}) \le {\left[O(1) N^{{\frac{1}{d_q}+\frac{\alpha}{2}-\frac{1}{2}}} \right]}^w
\end{align*}
The exponent verifies:
\[{{\frac{1}{d_q}+\frac{\alpha}{2}-\frac{1}{2}}} = \frac{1}{2}\left(\frac{2-d_q}{d_q}+\alpha\right) < 0\]
And since $p_N(w,\le N^{\alpha}) = 0$ when $w < d_c$, this implies:
\begin{align*}
\displaystyle\sum_{w=0}^{\lfloor N^{\alpha} \rfloor} p_N(w,\le N^{\alpha}) & \le \displaystyle\sum_{w=d_c}^{\infty} {\left(O(1)N^{{\frac{1}{d_q}+\frac{\alpha}{2}-\frac{1}{2}}}\right)}^w\\
& \le O(1) N^{({\frac{1}{d_q}+\frac{\alpha}{2}-\frac{1}{2}}) d_c}
\end{align*}
This proves that the first partial sum tends to $0$ as $N \rightarrow \infty$.
\end{proof}

Proof of Lemma \ref{Second poly}:
\begin{proof}{\bf Second partial sum, poly case}
We use:
\[p_N(w,d) \le \frac{a^{\otimes N}(w)\ a_{N_{in}}(w,d)}{{(|{\cal{P}}| - 1)}^w \binom{N n_{out}}{w}}\]
with the same bound as before derived from (1E), valid because $d_q \ge 2$:
\begin{align*}
a^{\otimes N}(w) & \le O(1)^w \left(\frac{N}{w}\right)^{\frac{w}{d_q}} & \mathrm{if}\ w \ge d_c \\
& = 0 & \mathrm{if} w < d_c
\end{align*}
and with the bound (2I), valid because $\bar{C_{in}}$ is totally recursive:
\[ a_{N_{in}}(w, d) \le O(1)^d \frac{N_{in}^{\frac{d}{2}} (w+d)^{\frac{d}{2}}}{d^d}\]
Since $N_{in} = O(N)$ and $w > d$, this last bound can be written in the following form:
\[a_{N_{in}}(w,d) \le O(1)^w \frac{N^{\frac{d}{2}} w^{\frac{d}{2}}}{d^d}\]
Let us also lower bound the denominator as previously by:
\[\left(\frac{N}{w}\right)^w\]
As a consequence:
\begin{align*}
p_N(w,d) & \le {O(1)}^w \left(\frac{N}{w}\right)^{\frac{w}{d_q}} \frac{N^{\frac{d}{2}} w^{\frac{d}{2}}}{d^d} \left(\frac{N}{w}\right)^{-w} \\
& \le {\left[O(1){\left(\frac{N}{w}\right)}^{\frac{1}{d_q}-1}\right]}^w \frac{N^{\frac{d}{2}} w^{\frac{d}{2}}}{d^d}
\end{align*}
Let us upper bound $O(1)$ by a constant $a > 0$, and take the logarithm of the last inequality:
\[\log(p_N(w,d)) \le f(w,d) = w(\log a + (1-1/{d_q})(\log w - \log N)) + \frac{d}{2}(\log N + \log w - 2 \log d)\]
Let us show that there exists $x$ sufficiently small such that, for $N$ sufficiently large, the maximum of $f(w,d)$ in the domain $(w,d) \in [\![\lfloor N^{\alpha} \rfloor,\lfloor x N \rfloor - 1]\!] \times [\![0,\lfloor N^{\alpha} \rfloor]\!]$ is reached when $w = d = \lfloor N^{\alpha} \rfloor$. First, consider the derivative of $f$ with respect to $w$:
\[\log a + (1-1/{d_q})(\log w - \log N + 1) + \frac{d}{2}\frac{1}{w}\]
Since $\frac{d}{2}\frac{1}{w} \le \frac{1}{2}$ and $w \le x N$, this is upper bounded by:
\[\log a + (1-1/{d_q})(\log x + 1) + \frac{1}{2}\]
For a sufficiently small $x$, this upper bound is negative. Indeed, when $x \rightarrow 0$, the above expression is dominated by the term:
\[(1-1/{d_q})\log x\]
which tends to $-\infty$. This shows that, provided $x$ is small enough, $f(w,d)$ is maximum when $w = \lfloor N^{\alpha} \rfloor$. Now, suppose that $w = \lfloor N^{\alpha} \rfloor$, and consider the derivative of $f$ with respect to $d$:
\[\frac{1}{2}(\log N + \log \lfloor N^{\alpha} \rfloor - 2 \log d - 2)\]
This is a decreasing function of $d$. Thus the minimum of this derivative is reached at $d = \lfloor N^{\alpha} \rfloor$ and is equal to:
\[\frac{1}{2}(\log N - \log \lfloor N^{\alpha} \rfloor - 2)\]
When $N \rightarrow \infty$ this minimum is equivalent to:
\[\frac{1 - \alpha}{2}\log N\]
Since $\alpha < 1$, for a sufficiently large $N$ this minimum is positive. Consequently $f(w,d)$ is maximum when $w = d = \lfloor N^{\alpha} \rfloor$. The maximum of $f$ in the domain is equal to:
\begin{align*}
f(\lfloor N^{\alpha} \rfloor,\lfloor N^{\alpha} \rfloor) = & \lfloor N^{\alpha} \rfloor(\log a + (1-1/{d_q})(\log \lfloor N^{\alpha} \rfloor - \log N)) + \frac{\lfloor N^{\alpha} \rfloor}{2}(\log N - \log \lfloor N^{\alpha} \rfloor)\\
= & \lfloor N^{\alpha} \rfloor(\log a + (1/2-1/{d_q})(\log \lfloor N^{\alpha} \rfloor - \log N))
\end{align*}
When $N \rightarrow \infty$ this maximum is equivalent to:
\begin{align*}
f(\lfloor N^{\alpha} \rfloor,\lfloor N^{\alpha} \rfloor) \sim - (1/2-1/{d_q})(1 - \alpha) N^{\alpha} \log N
\end{align*}
Now in the sum:
\[\displaystyle\sum_{d=0}^{\lfloor N^{\alpha} \rfloor}\displaystyle\sum_{w= \lfloor N^{\alpha} \rfloor + 1}^{\lfloor x N \rfloor - 1} p_N(w,d)\]
the number of terms is upper bounded by $x N^2$. Thus the logarithm of the sum is upper bounded by:
\[\log x + 2 \log N + f(\lfloor N^{\alpha} \rfloor,\lfloor N^{\alpha} \rfloor) \sim - (1/2-1/{d_q})(1 - \alpha) N^{\alpha} \log N\]
This equivalent tends to $- \infty$, which proves that the second partial sum tends to $0$ if we chose a sufficiently small $x$.
\end{proof}

Proof of Lemma \ref{Third poly}:
\begin{proof}{\bf Third partial sum, poly case}
Again we use:
\[p_N(w,d) \le \frac{a^{\otimes N}(w)\ a_{N_{in}}(w,d)}{{(|{\cal{P}}| - 1)}^w \binom{N n_{out}}{w}}\]
According to the bound (2E), valid because $|{\cal{P}}| > 2$ and $d_q \ge 2$, there exists a constant $c \in ]0,1[$ such that:
\[a^{\otimes N}(w) \le c^w (|{\cal{P}}| - 1)^w \binom{N n_{out}}{w}\]
Let us also use the bound (2I), valid because $\bar{C_{in}}$ is totally recursive:
\[ a_{N_{in}}(w, d) \le O(1)^d \frac{N_{in}^{\frac{d}{2}} (w+d)^{\frac{d}{2}}}{d^d}\]
Since $N_{in} = O(N)$, $d \le N$ and $w \le N n_{out}$, this bound gives:
\[a_{N_{in}}(w,d) \le O(1)^d N^d\]
With these two bounds we obtain:
\[p_N(w,d) \le c^w O(1)^d N^d\]
Thus:
\begin{align*}
\log(p'_N(w,d)) & \le w \log c + d (O(1) + \log N)\\
& \le \lfloor x N \rfloor \log c + N^{\alpha} (O(1) + \log N)
\end{align*}
where the last line is obtained using $\log c < 0$, $w \ge \lfloor x N \rfloor$ and $d \le N^{\alpha}$.
Since the third partial sum contains less than $n_{out} N^2$ terms, its logarithm is upper bounded by:
\[\log n_{out} + 2 \log N + \log(p'_N(w,d)) \le \log n_{out} + 2 \log N +  \lfloor x N \rfloor \log c + N^{\alpha} (O(1) + \log N)\]
Since $\alpha < 1$, this upper bound is equivalent to:
\[x N \log c\]
and this proves that the sum tends to $0$ when $N \rightarrow \infty$.
\end{proof}

\subsection*{B.2: Sub-logarithmic case}

Proof of Lemma \ref{First sublog}
\begin{proof}{\bf First partial sum, sub-log case}
Let $d = \alpha \mathrm{llog}\ N$. We use:
\[p_N(w,\le d) \le \frac{a^{\otimes N}(w)\ a_{N_{in}}(w,\le d)}{{(|{\cal{P}}| - 1)}^w \binom{N n_{out}}{w}}\]
Let us use this result from the bound (1E), valid because $d_q \ge 2$:
\begin{align*}
a^{\otimes N}(w) & \le O(1)^w \left(\frac{N}{w}\right)^{\frac{w-d_c}{2}+1}
\end{align*}
and the bound (1I), valid because $C_{in}$ is recursive:
\[ a_{N_{in}}(w,\le d) \le O(1)^w \frac{N_{in}^{\frac{w}{2}} (w+d)^{\frac{w}{2}}}{w^w}\]
Since $w \le d$ and $N_{in} = O(N)$ this bound gives:
\[ a_{N_{in}}(w,\le d) \le O(1)^w \frac{N^{\frac{w}{2}} d^{\frac{w}{2}}}{w^w}\]
Let us also lower bound the denominator:
\[{(|{\cal{P}}| - 1)}^w \binom{N n_{out}}{w} \ge \left(\frac{N}{w}\right)^w\]
By combining these three bounds we get:
\begin{align*}
p_N(w,\le d) \le & O(1)^w \left(\frac{N}{w}\right)^{\frac{w-d_c}{2}+1} \frac{N^{\frac{w}{2}} d^{\frac{w}{2}}}{w^w} \left(\frac{N}{w}\right)^{-w}\\
\le & O(1)^w N^{\frac{w-d_c}{2}+1} N^{\frac{w}{2}} d^{\frac{w}{2}} N^{-w}\\
\le & [O(1)d]^{\frac{w}{2}} N^{\frac{-d_c}{2}+1}
\end{align*}
The last bound is an increasing function of $w$. Applied at $w = \alpha \mathrm{llog}\ N$, it is thus a bound for each term of the first partial sum:
\[\displaystyle\sum_{w=0}^{\lfloor \alpha\ \mathrm{llog}\ N \rfloor} p_N(w, \le \alpha\ \mathrm{llog}\ N) \le (\alpha \mathrm{llog}\ N + 1) {\left[O(1)\alpha\ \mathrm{llog}\ N\right]}^{\frac{\alpha\ \mathrm{llog}\ N}{2}} N^{-\frac{d_c}{2}+1}\]
Now, notice that:
\begin{align*}
{O(1)}^{\frac{\alpha\ \mathrm{llog}\ N}{2}} = \exp \left(O(1)\frac{\alpha\ \log N}{2 \log(\log N)}\right) = N^{\frac{O(1)}{\log(\log N)}}
\end{align*}
And using Proposition  \ref{llog^llog}:
\[{\left[\alpha\ \mathrm{llog}\ N\right]}^{\frac{\alpha\ \mathrm{llog}\ N}{2}} \le N^{\frac{\alpha}{2}}\]
Mutliplying these terms we get:
\[\displaystyle\sum_{w=0}^{\lfloor \alpha\ \mathrm{llog}\ N \rfloor} p_N(w, \le \alpha\ \mathrm{llog}\ N) \le (\alpha \mathrm{llog}\ N + 1) N^{\frac{O(1)}{\log(\log N)}} N^{\frac{\alpha - d_c}{2}+1}\]
Since $\frac{\alpha - d_c}{2}+1 < 0$ the first partial sum tends to $0$ when $N \rightarrow \infty$.
\end{proof}

Proof of Lemma \ref{Second sublog}
\begin{proof} {\bf Second partial sum, sub-log case}
We use:
\[p_N(w,d) \le \frac{a^{\otimes N}(w)\ a_{N_{in}}(w,d)}{{(|{\cal{P}}| - 1)}^w \binom{N n_{out}}{w}}\]
Let us apply this result from the bound (1E), valid because $d_q \ge 2$:
\begin{align*}
a^{\otimes N}(w) & \le O(1)^w \left(\frac{N}{w}\right)^{\frac{w-d_c}{2}+1}
\end{align*}
and the bound (2I), valid because $\bar{C_{in}}$ is totally recursive:
\[a_{N_{in}}(w,d) \le O(1)^d \frac{{N_{in}}^{\frac{d}{2}} (w+d)^{\frac{d}{2}}}{d^d}\]
Using the facts that $N_{in} = O(N)$ and $w > d$, this bound gives:
\[a_{N_{in}}(w,d) \le O(1)^w \frac{N^{\frac{w}{2}} w^{\frac{d}{2}}}{d^d}\]
We also lower bound the denominator:
\[{(|{\cal{P}}| - 1)}^w \binom{N n_{out}}{w} \ge \left(\frac{N}{w}\right)^w\]
Combining these bounds gives:
\begin{align*}
p_N(w,d) \le & O(1)^w \left(\frac{N}{w}\right)^{\frac{w-d_c}{2}+1} \frac{N^{\frac{w}{2}} w^{\frac{d}{2}}}{d^d} \left(\frac{N}{w}\right)^{-w}\\
 \le & N^{-\frac{d_c}{2}+1} {\left[O(1){\left(\frac{N}{w}\right)}^{-\frac{1}{2}}\right]}^w \frac{N^{\frac{d}{2}} w^{\frac{d}{2}}}{d^d}
\end{align*}
Let us upper bound $O(1)$ by a constant $a > 0$; then consider the logarithm $f(w,d)$ of the right hand expression apart from the first factor:
\begin{align*}
f(w,d) & = \log \left( {\left[a{\left(\frac{N}{w}\right)}^{-\frac{1}{2}}\right]}^w \frac{N^{\frac{d}{2}} w^{\frac{d}{2}}}{d^d} \right)\\
& = w(\log a + 1/2(\log w - \log N)) + \frac{d}{2}(\log N + \log w - 2 \log d)
\end{align*}
Let us show that there exists $x$ sufficiently small such that, for $N$ sufficiently large, the maximum of $f(w,d)$ in the domain $(w,d) \in [\![\lfloor \alpha\ \mathrm{llog}\ N \rfloor;\lfloor x N \rfloor - 1]\!] \times [\![0;\lfloor \alpha\ \mathrm{llog}\ N \rfloor]\!]$ is reached when $w = d = \lfloor \alpha\ \mathrm{llog}\ N \rfloor$.
The derivative of $f$ with respect to $w$ is equal to:
\[\log a + 1/2(\log w -\log N + 1) + \frac{d}{2 w} \le \log a + 1/2(\log x + 1) + \frac{1}{2}\]
where the inequality comes from the facts that $w \le x N$ and $d \le w$. For a sufficiently small $x$, the right hand side (and thus the partial derivative) is negative. This shows that the maximum of $f(w,d)$ is reached if $w = \lfloor \alpha\ \mathrm{llog}\ N \rfloor$. Moreover, there exists $r \in ]0;1[$ such that the derivative of $f$ with respect to $w$ is upper bounded by $\log r$. By integrating this inequality from $\lfloor \alpha\ \mathrm{llog}\ N \rfloor$ to $w$ we get the following inequality:
\begin{align}
f(w,d) \le f(\lfloor \alpha\ \mathrm{llog}\ N \rfloor,d) + \log r \left(w - \lfloor \alpha\ \mathrm{llog}\ N \rfloor \right)
\label{pente}
\end{align}
Now suppose that $w = \lfloor \alpha\ \mathrm{llog}\ N \rfloor$, and consider the derivative of $f$ with respect to $d$:
\[ \frac{1}{2}(\log N + \log \lfloor \alpha\ \mathrm{llog}\ N \rfloor - 2 \log d - 2 )\]
This derivative is a decreasing function of $d$; its minimum is reached when $d = \lfloor \alpha\ \mathrm{llog}\ N \rfloor$ where it is equal to:
\[ \frac{1}{2}(\log N - \log \lfloor \alpha\ \mathrm{llog}\ N \rfloor - 2 )\]
which is positive for $N$ large enough. Thus $f(w,d)$ is maximum when $w = d = \lfloor \alpha\ \mathrm{llog}\ N \rfloor$. The maximum of $f$ in the domain is then equal to:
\begin{align*}
& \lfloor \alpha\ \mathrm{llog}\ N \rfloor(\log a + 1/2(\log \lfloor \alpha\ \mathrm{llog}\ N \rfloor - \log N)) + \frac{\lfloor \alpha\ \mathrm{llog}\ N \rfloor}{2}(\log N - \log \lfloor \alpha\ \mathrm{llog}\ N \rfloor)\\
& = \lfloor \alpha\ \mathrm{llog}\ N \rfloor \log a = O(1)\ \mathrm{llog}\ N
\end{align*}
Thus by inequality \ref{pente}, we get for all $(w,d)$ in the domain:
\[f(w,d) \le O(1)\ \mathrm{llog}\ N + \log r \left(w - \lfloor \alpha\ \mathrm{llog}\ N \rfloor \right)\]
This proves in turn that :
\[p_N(w,d) \le N^{-\frac{d_c}{2}+1} N^{\frac{O(1)}{\log(\log N)}} r^{w - \lfloor \alpha\ \mathrm{llog}\ N \rfloor}\]
Thus for each $d \in [\![0;\lfloor \alpha\ \mathrm{llog}\ N \rfloor]\!]$, since $r < 1$:
\begin{align*}
\displaystyle\sum_{w= \lfloor \alpha\ \mathrm{llog}\ N \rfloor + 1}^{\lfloor x N \rfloor - 1} p_N(w,d) \le & \displaystyle\sum_{w= \lfloor \alpha\ \mathrm{llog}\ N \rfloor}^{\infty} N^{-\frac{d_c}{2}+1} N^{\frac{O(1)}{\log(\log N)}} r^{w - \lfloor \alpha\ \mathrm{llog}\ N \rfloor}\\
\le & O(1) N^{-\frac{d_c}{2}+1} N^{\frac{O(1)}{\log(\log N)}}
\end{align*}
Consequently, by summing over $d$, the second partial sum is upper bounded by:
\[O(1)\ (\mathrm{llog}\ N + 1) N^{-\frac{d_c}{2}+1} N^{\frac{O(1)}{\log(\log N)}}\]
which tends to $0$ when $N \rightarrow \infty$ since $d_c > 2$.
\end{proof}

Proof of Lemma \ref{Third sublog} {\bf Third partial sum, sub-log case}
\begin{proof}
The proof is exactly the same as the third partial sum in the polynomial case, by changing every occurrence of $N^{\alpha}$ by $\alpha\ \mathrm{llog}\ N$.
\end{proof}


\end{document}